\documentclass[journal,comsoc]{IEEEtran}
\usepackage[utf8]{inputenc}
\usepackage[T1]{fontenc}
\usepackage{url,cite}
\usepackage[cmex10]{amsmath}
\usepackage{amssymb,mathtools, amsthm}
\usepackage{graphicx}
\usepackage{multirow}
\usepackage[font = scriptsize, labelfont = bf, labelsep = period, tableposition = top, singlelinecheck = false, justification = centering]{caption}

\usepackage{enumerate}

\usepackage{tikz}
\usetikzlibrary{tikzmark, shapes, fit,backgrounds}
\usetikzlibrary{matrix,decorations.pathreplacing,angles,quotes}
\usetikzlibrary{patterns, calc, positioning}
\tikzstyle{block} = [rectangle, draw, 
    minimum width=4em, text centered, rounded corners, minimum height=1em]
\tikzstyle{line} = [draw, -latex]

\newtheorem{lemm}{Lemma}[section]
\newtheorem{defi}{Definition}[section]

\newtheorem{theo}{Theorem}[section]
\newtheorem{prop}{Proposition}[section]

\newtheorem{prot}{Protocol}[section]
\newtheorem{nota}{Notation}[section]
\newtheorem{coro}{Corollary}[section]

\theoremstyle{definition}
\newtheorem{remark}{Remark}
\newtheorem*{remark*}{Remark}

\newif\ifcomment
\commenttrue

\newif\ifcommentLater
\commentLaterfalse

\newcommand{\ie}{\emph{i.e.}, }
\definecolor{dartmouthgreen}{rgb}{0.05, 0.5, 0.06}
\newcommand{\commentOut}[1]{}

\makeatletter
\renewcommand*\env@matrix[1][*\c@MaxMatrixCols c]{%
  \hskip -\arraycolsep
  \let\@ifnextchar\new@ifnextchar
  \array{#1}}
\makeatother

\DeclareMathOperator{\diagonal}{diag}
\DeclareMathOperator{\trace}{Tr}

\DeclareMathOperator{\spacespan}{span}

\DeclareMathOperator{\lcm}{lcm}
\DeclareMathOperator{\id}{{\mathrm{id}}}
\DeclareMathOperator{\ttr}{tr}
\DeclareMathOperator{\image}{Im}

\DeclarePairedDelimiter\px{\{}{\}}
\DeclarePairedDelimiter\paren{(}{)}

\DeclarePairedDelimiter\ceil{\lceil}{\rceil}

\DeclarePairedDelimiterX\symp[2]{\langle}{\rangle_\mathbb{S}}{#1 , #2}

\newcommand{\cA}{\mathcal{A}}
\newcommand{\cB}{\mathcal{B}}
\newcommand{\cC}{\mathcal{C}}
\newcommand{\cD}{\mathcal{D}}
\newcommand{\cE}{\mathcal{E}}

\newcommand{\cH}{\mathcal{H}}

\newcommand{\cJ}{\mathcal{J}}

\newcommand{\cL}{\mathcal{L}}
\newcommand{\cM}{\mathcal{M}}

\newcommand{\cQ}{\mathcal{Q}}
\newcommand{\cR}{\mathcal{R}}
\newcommand{\cS}{\mathcal{S}}
\newcommand{\cT}{\mathcal{T}}

\newcommand{\cV}{\mathcal{V}}
\newcommand{\cW}{\mathcal{W}}

\newcommand{\bA}{\mathbf{A}}
\newcommand{\bB}{\mathbf{B}}

\newcommand{\bE}{\mathbf{E}}
\newcommand{\bF}{\mathbf{F}}
\newcommand{\bG}{\mathbf{G}}
\newcommand{\bH}{\mathbf{H}}
\newcommand{\bI}{\mathbf{I}}
\newcommand{\bJ}{\mathbf{J}}

\newcommand{\bM}{\mathbf{M}}
\newcommand{\bN}{\mathbf{N}}

\newcommand{\bP}{\mathbf{P}}
\newcommand{\bQ}{\mathbf{Q}}

\newcommand{\bT}{\mathbf{T}}

\newcommand{\bX}{\mathbf{X}}
\newcommand{\bY}{\mathbf{Y}}
\newcommand{\bZ}{\mathbf{Z}}

\newcommand{\bc}{\mathbf{c}}
\newcommand{\bd}{\mathbf{d}}
\newcommand{\be}{\mathbf{e}}

\newcommand{\bm}{\mathbf{m}}

\newcommand{\bo}{\mathbf{o}}

\newcommand{\br}{\mathbf{r}}

\newcommand{\bw}{\mathbf{w}}

\def\sn{\mathsf{n}}
\def\sd{\mathsf{d}}
\def\st{\mathsf{t}}
\def\sk{\mathsf{k}}
\def\sm{\mathsf{m}}
\def\sX{\mathsf{X}}
\def\sW{\mathsf{W}}
\def\sZ{\mathsf{Z}}

\def\a{\alpha}
\def\b{\beta}

\def\d{\delta}

\def\s{\sigma}

\def\diag#1{\diagonal \paren{#1}}			
\def\Tr#1{\trace \paren{#1}}					
\def\Span#1{\spacespan\px{#1}}
\def\cb#1{| #1 \rangle}						
\def\rb#1{\langle #1 |}						
\def\ppmatrix#1{\begin{pmatrix}#1\end{pmatrix}} 
\def\imag#1{\image #1}			

\def\C{\mathbb{C}}
\def\F{\mathbb{F}}
\def\Fq{\F_{q}}
\def\N{\mathbb{N}}
\def\sinit{\s_\mathrm{init}}

\def\tK{K}
\def\tk{\iota}

\newcommand{\vle}{\rotatebox[origin=c]{-90}{$\le$}}

\DeclareMathOperator*{\argmax}{\arg\!\max}
\usepackage{enumitem}

\usepackage[innermargin=0.545in,outermargin=0.545in,top=0.6in,bottom=.6in]{geometry}
\begin{document}

\title{On the Capacity of Quantum Private Information Retrieval from MDS-Coded and Colluding Servers}

\author{\IEEEauthorblockN{Matteo Allaix,~\IEEEmembership{Student Member,~IEEE}, Seunghoan Song,~\IEEEmembership{Member,~IEEE}, Lukas Holzbaur, \IEEEmembership{Graduate Student Member,~IEEE}, Tefjol Pllaha, Masahito Hayashi,~\IEEEmembership{Fellow,~IEEE}, Camilla Hollanti,~\IEEEmembership{Member,~IEEE}}

\thanks{Partial results have been published at ISIT 2021 \cite{allaix2021quantum}. 
C.~Hollanti and M.~Allaix were supported by the Academy of Finland, under Grants No. 318937 and 336005.
S.~Song was supported by JSPS Grant-in-Aid for JSPS Fellows No. JP20J11484.
L.~Holzbaur was supported by the German Research Foundation (Deutsche Forschungsgemeinschaft, DFG) under Grant No.~WA 3907/1-1. 
M.~Hayashi was supported in part by Guangdong Provincial Key Laboratory (Grant No. 2019B121203002).
\emph{(The first two authors contributed equally to this work.)}
}
\thanks{M. Allaix and C. Hollanti are with the Department of Mathematics and System Analysis, Aalto University, Espoo, Finland
(e-mails: \{matteo.allaix, camilla.hollanti\}@aalto.fi).} 
\thanks{S. Song is with Graduate school of Mathematics, Nagoya University, Nagoya, 464-8602, Japan
(e-mail: m17021a@math.nagoya-u.ac.jp).} 
\thanks{L. Holzbaur is with the Institute for Communications Engineering, Technical University of Munich, Germany (e-mail: lukas.holzbaur@tum.de).}
\thanks{T. Pllaha is with the Department of Mathematics, University of Nebraska, Lincoln, USA (email: tefjol.pllaha@unl.edu).}
\thanks{M. Hayashi is with 
Shenzhen Institute for Quantum Science and Engineering, Southern University of Science and Technology,
Shenzhen, 518055, China,
Guangdong Provincial Key Laboratory of Quantum Science and Engineering,
Southern University of Science and Technology, Shenzhen 518055, China,
and Graduate School of Mathematics, Nagoya University, Nagoya, 464-8602, Japan
(e-mail:hayashi@sustech.edu.cn).}
}

\maketitle

\begin{abstract}
In quantum private information retrieval (QPIR), a user retrieves a classical file from multiple servers by downloading quantum systems without revealing the identity of the file. The QPIR capacity is the maximal achievable ratio of the retrieved file size to the total download size. In this paper, the capacity of QPIR from MDS-coded and colluding servers is studied for the first time. Two general classes of QPIR, called stabilizer QPIR and dimension-squared QPIR induced from classical strongly linear PIR are defined, and the related QPIR capacities are derived. For the non-colluding case, the general QPIR capacity is derived when the number of files goes to infinity. A general statement on the converse bound for QPIR with coded and colluding servers is derived showing that the capacities of stabilizer QPIR and dimension-squared QPIR induced from any class of PIR are upper bounded by twice the classical capacity of the respective PIR class. The proposed capacity-achieving scheme combines the star-product scheme by Freij-Hollanti \emph{et al.} and the stabilizer QPIR scheme by Song \emph{et al.} by employing (weakly) self-dual Reed--Solomon codes.
\end{abstract}

\section{Introduction}

With the amount of data stored in distributed storage systems steadily increasing, the demand for user privacy has surged in recent years. One notion that has received considerable attention is private information retrieval (PIR), where the user's goal is to access a file of a (distributed) storage system without revealing the identity (index) of this desired file. In their seminal work Chor et al. \cite{chor1995private} introduced the concept of PIR from multiple non-colluding servers, each storing a copy of every file. More recently, the capacity, \ie the highest achievable rate, for this setting \cite{sun2017replicated} was derived, which led to similar derivations in more general settings admitting for colluding servers \cite{sun2017capacity}, coded storage \cite{banawan2018capacity}, and symmetric privacy \cite{sun2018capacity,wang2017linear}. While the capacity of PIR from coded storage with colluding servers remains an open problem, some progress was made in \cite{Sun2018conjecture,holzbaur2019capacity,Holzbaur2019ITW}. Among other things, \cite{holzbaur2019capacity,Holzbaur2019ITW} introduce the practical notion of strongly linear PIR. Informally, this class is given by PIR schemes where both the computation of the server responses and the decoding of the desired file from these responses is achieved by applying linear functions. The capacity of this class of schemes coincides with a conjecture on the asymptotic (in the number of files) capacity for this setting \cite{tajeddine2019private} and is known to be achievable by schemes with requiring only small subpacketization, such as the star-product scheme of \cite{freij2017private}.

Quantum PIR (QPIR) considers accomplishing the PIR task with quantum communication between the user and the servers \cite{kerenidis03,kerenidis2004quantum,legall2011quantum,olejnik2011,baumeler2015,kerenidis2016,aharonov2019,Kon2020}.
Following the study on the classical PIR capacity \cite{sun2017capacity}, 
   the papers \cite{song2019capacity,song2019allbutone,song2020colluding,allaix2020quantum} considered the capacity of QPIR and quantum symmetric PIR (QSPIR), where the user obtains no other information than the desired file in addition to the requirements of PIR.
The QPIR schemes in \cite{song2019capacity,song2019allbutone,song2020colluding,allaix2020quantum} are conducted by the following procedure:
    a user uploads classical queries;
    multiple servers sharing entanglement
        apply quantum operations on their quantum systems depending on the queries and the files
        and
        respond quantum systems to the user;
    the user finally retrieves the desired file by quantum measurement on the responded systems.
When each of the $\sn$ servers stores a copy of every file, the QPIR/QSPIR capacity with multiple non-colluding servers \cite{song2019capacity} and $\st$ colluding servers \cite{song2020colluding} are proved to be $1$ and $\min \{ 1, 2(\sn-\st)/\sn\}$, respectively.
On the other hand, when the files are stored in a  distributed storage system coded by an $[\sn,\sk]$ maximum distance separable (MDS) code,
    QSPIR schemes with colluding servers are constructed \cite{allaix2020quantum}, but the result was limited to the case $\st +\sk = \sn $.

\subsection{Contributions}
\label{sec:contrib}

\begin{table*}[t]
    \centering
    \caption{Known asymptotic ($m\rightarrow \infty$) capacity results with $\sn$ servers.
    The result in \textcolor{red}{red} is a conjecture in its full generality \cite{freij2017private}, but shown to hold for strongly linear \cite{Holzbaur2019ITW} and full support rank \cite{holzbaur2019capacity} PIR. A scheme achieving that rate was proposed in \cite{freij2017private}. The results in \textcolor{dartmouthgreen}{green} are proved in this paper for strongly linear PIR.
    }
    \begin{tabular}{|l|cc|cc|cc|}
    \hline
    \textsc{Capacities} & PIR & ref. & SPIR & ref. & QPIR & ref. \\
    \hline
    Replicated storage, & \multirow{2}{*}{$1-\frac{1}{\sn}$} & \multirow{2}{*}{\cite{sun2017replicated}} & \multirow{2}{*}{$1-\frac{1}{\sn}$}  & \multirow{2}{*}{\cite{sun2018capacity}} & \multirow{2}{*}{1} & \multirow{2}{*}{\cite{song2019capacity}} \\
    no collusion & & & & & & \\
    \hline
    Replicated storage, & \multirow{2}{*}{$1-\frac{\st}{\sn}$} & \multirow{2}{*}{\cite{sun2017capacity}} & \multirow{2}{*}{$1-\frac{\st}{\sn}$}  & \multirow{2}{*}{\cite{wang2017colluding}} & \multirow{2}{*}{$\min\{1,\frac{2(\sn-\st)}{\sn}\}$} & \multirow{2}{*}{\cite{song2020colluding}} \\
    $\st$-collusion & & & & & & \\
    \hline
    $[\sn,\sk]$-MDS coded & \multirow{2}{*}{$1-\frac{\sk}{\sn}$} & \multirow{2}{*}{\cite{banawan2018capacity}} & \multirow{2}{*}{$1-\frac{\sk}{\sn}$}  & \multirow{2}{*}{\cite{wang2017linear}} & \multirow{2}{*}{\textcolor{dartmouthgreen}{$\min\{1,\frac{2(\sn-\sk)}{\sn}\}$}} & \multirow{2}{*}{--} \\
    storage, no collusion & & & & & & \\
    \hline
    $[\sn,\sk]$-MDS coded & \multirow{2}{*}{\textcolor{red}{$1-\frac{\sk+\st-1}{\sn}$}} & \multirow{2}{*}{\cite{freij2017private}} & \multirow{2}{*}{$1-\frac{\sk+\st-1}{\sn}$}  & \multirow{2}{*}{\cite{wang2017linear}
    } & \multirow{2}{*}{\textcolor{dartmouthgreen}{$\min\{1,\frac{2(\sn-\sk-\st+1)}{\sn}\}$}} & \multirow{2}{*}{--} \\
    storage, $\st$-collusion & & & & & & \\
    \hline
    \end{tabular}
    \label{tab:Capacities}
\end{table*}

As a generalization of \cite{allaix2020quantum}, we study the QPIR/QSPIR capacity from $[\sn,\sk]$ MDS coded storage with $\st$ colluding servers for any $\st+\sk\le \sn$.
Since the capacity of this setting is even unsolved for the classical case, similar to \cite{holzbaur2019capacity,Holzbaur2019ITW}, we define two new classes of QPIR, which include the existing QPIR schemes [21]–[24], and derive the capacity for these classes.
The first class is {\em stabilizer QPIR induced from classical PIR}.
Stabilizer QPIR is a class of QPIR that naturally imports linear PIR schemes in quantum settings while doubling the PIR rate.
More specifically, the user and the servers simulate the classical PIR scheme, except that
 	the servers' prior entangled state is a state in a stabilizer code
	and
 	the servers apply Pauli $\sX$ and $\sZ$ operations on each quantum system depending on the answers of the classical PIR.
The second class is {\em dimension-squared QPIR}, which is a broader class of QPIR that includes stabilizer QPIR.
Whereas the stabilizer QPIR is defined with restrictions on the encoding, decoding, and shared entanglement, dimension-squared QPIR is defined only with restriction on dimensions of the answered quantum systems, which is a sufficient condition for our converse proof.
Similar to the stabilizer QPIR, dimension-squared QPIR can also be induced from classical PIR and the existing QPIR schemes \cite{song2019capacity,song2019allbutone,song2020colluding,allaix2020quantum} are dimension-squared QPIR induced from strongly linear PIR.

For stabilizer QPIR and dimension-squared QPIR induced from strongly-linear PIR, we prove that the asymptotic QPIR/QSPIR capacities with MDS-coded and colluding servers are $\min \{ 1, 2(\sn-\sk-\st+1)/\sn\}$.
Furthermore, for non-colluding case $\st = 1$, we prove that the general asymptotic QPIR/QSPIR capacity is $\min \{ 1, 2(\sn-\sk)/\sn\}$.
The derived quantum capacities double the classical asymptotic capacities of PIR and SPIR, as compared in Table~\ref{tab:Capacities}.

The capacity achieving scheme is based on the strongly-linear star-product scheme  of \cite{freij2017private} for classical PIR from MDS-coded storage and the QPIR scheme of \cite{song2020colluding} for replicated storage, both in the presence of $\st$ colluding servers. A generalization of these schemes, which employs (weakly) self-dual Generalized Reed--Solomon (GRS) codes, results in the first known QPIR scheme from MDS-coded storage in the considered setting. The scheme is non-trivial for two main reasons. First, the chosen codes must behave well with the star-product operation: one example is the polynomial-based codes class, that includes GRS codes. This requirement comes from the classical PIR scheme described in \cite{freij2017private}. Second, the star-product of the storage code and the query code must be a (weakly) self-dual code in order to employ the stabilizer formalism and get the advantage of quantum communication. To the best of our knowledge the combination of these two properties was not considered in previous literature. In this paper, we prove that for any given GRS storage code we can find a GRS query code such that their star-product is a (weakly) self-dual code.

The converse bounds are proved separately for the colluding and non-colluding cases.
First, the converse for colluding case is derived generally for any PIR classes.
Namely, when the classical capacity of any PIR class is $C$,
we prove that the rates of stabilizer QPIR and dimension-squared QPIR induced from the same class of PIR are upper bounded by  $\min\{1,2C\}$.
Then, from the capacity of strongly linear PIR for coded and colluding servers $(\sn-\sk-\st+1)/\sn$ \cite{holzbaur2019capacity,Holzbaur2019ITW}, we obtain our converse bound for colluding case.
Second, the converse for non-colluding case is proved for general QPIR schemes with the following idea.
We prove that the $\sk$ servers obtain negligible information of the user's information.
Combining this fact and the entanglement-assisted classical capacity \cite{bennett1999}, we prove that the desired converse bound $C \leq \min \{ 1, 2(\sn-\sk)/\sn\}$.

Similar to the existing multi-server QPIR studies \cite{song2019capacity,song2019allbutone,song2020colluding,allaix2020quantum},
the communication model in this paper is classical query and quantum answers with entanglement.
This model is the hybrid model of classical and quantum communication for classical file retrieval.
Compared to the non-quantum model, our main theorem implies that the capacity doubles only with the one-way quantum communication from the servers to the user.
On the other hand, compared to the purely quantum model, which allows quantum queries, our model has three practical advantages.
First, since the quantum communication is hard to be implemented with the current technology, our one-way communication model is a more realizable model than the two-way quantum communication.
Second, in our scheme, most of the quantum resources and computations are operated by the servers, and the only quantum device required for the user is a fixed measurement apparatus.\footnote[1]{QPIR problem can also be considered for the retrieval of quantum states, i.e., QPIR with quantum storage. A part of authors discussed this problem in a recent paper \cite{QQPIR21}.}
The same kind of outsourcing also appears in the blind computation by measurement-based quantum computation \cite{PhysRevLett.115.220502}.
Third, since the storage is still classical, we can just employ quantum communication technology and quantum memory to double the rate of an already existing MDS-coded storage implementing a classical PIR scheme.

\subsection{Organization}

The remainder of the paper is organized as follows.
Section~\ref{sec:preliminaries} is a preliminary section for notation, linear codes and distributed data storage, quantrum information theory, and stabilizer formalism.
In Section~\ref{sec:PIR_def}, we formally define classical PIR, QPIR, and the related QPIR classes.
In Section~\ref{sec:main}, we present our main capacity results.
Our capacity-achieving QPIR scheme with MDS-coded storage and colluding servers is proposed in Section~\ref{sec:achieve}
and the converse bound is derived in Section~\ref{sec:converse}.
Section~\ref{sec:conclusion} is the conclusion of the paper.

\section{Preliminaries} \label{sec:preliminaries}

\subsection{Notation}

We denote by $[n]$ and $[n_1 : n_2]$ the sets $\px{1,2,\ldots,n}, n \in \N$ and $\px{n_1,n_1+1,\ldots,n_2}, n_1,n_2 \in \N$, respectively, and by $\Fq$ the finite field of $q$ elements. For a linear code of length $\sn$ and dimension $\sk$ over $\Fq$ we write $[\sn,\sk]$. For random variables $A_1,\ldots, A_{n}$, quantum systems $\cA_1,\ldots, \cA_{n}$ and a set $\cS \subset [n]$, we denote $A_{\cS} \coloneqq ( A_j \mid j \in \cS )$ and $\cA_{\cS} \coloneqq \bigotimes_{j\in\cS} \cA_j$. For a matrix $\bA$ we write $\bA^\top$ for its transpose and $\bA^\dagger$ for its conjugate transpose. The function $\d_{i,j}$ is the Kronecker delta and $\bI_\nu$ is the $\nu \times \nu$ identity matrix.
For an $n \times m$ matrix $\bA = (a_{ij})_{i\in[n],j\in[m]}$, $\cS_1\subset[n]$, and $\cS_2\subset[m]$,
    we denote $\bA^{\cS_1}_{\cS_2} = (a_{ij})_{i\in \cS_1,j\in\cS_2}$
    and 
    $\bA^{\cS_1} = (a_{ij})_{i\in \cS_1,j\in [m]}$, $\bA_{\cS_2} = (a_{ij})_{i\in [n],j\in \cS_2}$.
Throughout this paper, we use $\log$ for the logarithm to the base $2$.

\subsection{Linear codes and distributed data storage} \label{subsec:storage}

We consider a distributed storage system employing error/erasure correcting codes to protect against data loss.
To this end, let $\bX$ be an $\sm\beta \times \sk$ matrix containing $\sm$ files $\bX^i \in \Fq^{\b \times \sk},\ i\in [\sm]$. This matrix is encoded with a linear code $\cC$ of length $\sn$ and dimension $\sk$ over $\Fq$. The $\sm\beta \times \sn$ matrix of encoded files is given by $\bY = \bX \cdot \bG_\cC$, where $\bG_\cC \in \Fq^{\sk \times \sn}$ is the generator matrix of $\cC$. Server $s \in [\sn]$ stores the $s$-th column of $\bY$, which is denoted by $\bY_s$.

In this work we consider systems encoded with MDS codes.
A linear code $\cC$ is called an MDS code if any $\sk$ columns of the generator matrix $\bG_\cC$ are linearly independent. 
Since we consider a MDS coded data storage, we have the following properties. 
\begin{enumerate}
    \item The matrix $\bX^i$ can be recovered from any $\sk$ elements of $\{\bY^i_1, \ldots, \bY^i_{\sn}\}$ for any $i \in [\sm]$.
    
    \item Any $\sk$ columns of $\bY$ are linearly independent.
\end{enumerate}

\subsection{Preliminaries on quantum information theory}

In this subsection, we introduce the preliminaries on quantum information theory.
To be precise, we introduce quantum systems, states, operations, and measurements.
Further, after the introduction,
we explain the quantum information theory is a generalization of classical information theory.
For more details the reader is referred to \cite{NC00, Hay17}.

  A quantum system $\cH$ is represented by a finite dimensional complex vector space.
  Vectors in a quantum system are written with bra-ket notation as $|\psi\rangle\in\cH$ and their complex conjugates are as $\langle \psi |$.
  The {\em computational basis} of a $d$-dimensional quantum system $\cH$ is a fixed orthonormal basis written as $\{|0\rangle, \ldots, |d-1\rangle\}$.
  The composite system of multiple quantum systems $\cH_1, \ldots, \cH_n$ is represented by the tensor product $\cH_1\otimes \cdots \otimes \cH_n$.
  
  A state $\sigma$ on $\cH$ is represented by a positive-semidefinite matrix on $\cH$ with trace $1$, which is called a {\em density matrix}.
  When a density matrix $\sigma$ is a rank-one matrix, i.e., $\sigma=|\psi\rangle\langle\psi|$, the state is equivalently represented by a unit vector $|\psi\rangle$, called a {\em pure state}.
  When a state is not a pure state, the state is called a {\em mixed state}.
  On a composite system $\cH_1\otimes \cdots \otimes \cH_n$, a state is called {\em separable} if the state is written as $\sigma = \sum_{i} p_i \sigma_1\otimes \cdots \otimes \sigma_n$ with $p_i\ge 0$, $\sum_i p_i = 1$, and density matrices $\sigma_i$ for all $i$.
  A state on a composite system is called {\em entangled} if it is not a separable state.
  When the state on $\cH_1\otimes \cdots \otimes \cH_n$ is $\sigma$, the {\em reduced state} on $\cH_k$ is written as $\mathrm{Tr}_{k^c} \sigma$, where $\mathrm{Tr}_{k^c}$ is the partial trace over $\bigotimes_{i\neq k} \cH_i$.

  A quantum operation $\kappa$ from $\cH_1$ to $\cH_2$ is represented by {\em completely positive trace-preserving (CPTP) map} defined as follows.
  A linear map $\kappa$ from matrices on $\cH_1$ to matrices on $\cH_2$
   is called {\em completely positive} if for all positive integer $n$, the map $\kappa\otimes \id_{\mathbb{C}^n}$ maps positive-semidefinite matrices to positive-semidefinite matrices, where $\id_{\mathbb{C}^n}$ is the identity map over the matrices on $\id_{\mathbb{C}^n}$,
   and {\em trace-preserving} if $\Tr {\kappa(M)} = \Tr M$ for all matrices $M$ on $\cH_1$.
  A CPTP map $\kappa$ is called a {\em unitary map} if $\kappa(M) = U^\dagger M U$ with a unitary matrix $U$ on $\cH$.

  A measurement on a quantum system $\cH$ is represented by a set of positive-semidefinite matrices $\mathbf{M} = \{M_\omega\}_{\omega\in\Omega}$ on $\cH$ with $\sum_\omega M_\omega = I$, called a {\em positive operation-valued measure (POVM)}.
  When a POVM is performed on a state $\sigma$, the measurement outcome is $\omega$ with probability $\Tr {M_\omega \sigma M_\omega}$.
  If all elements of a POVM $\{M_\omega\}_{\omega\in\Omega}$ are orthogonal projections, the POVM is called {\em the projection-valued measure (PVM)}.

Classical information theory is included in the framework of quantum information theory in the following sense.
  A finite set $[0:d-1]$ 
  corresponds to a $d$-dimensional quantum system with computational basis $\{|0\rangle, \ldots, |d-1\rangle\}$.
  An instance $x\in[0:d-1]$ and a random variable $X$ with probability $\{p_x | x\in [0:d-1] \}$ correspond, respectively, to a pure state $|x\rangle$ and a mixed state $\sigma = \sum_{x\in[0:d-1]} p_x |x\rangle \langle x|$.
  A transition matrix $Q = (Q_{x,y})_{x\in[0:d-1], y\in[0:d'-1]}$, which satisfies $Q_{x,y} \in [0,1]$ and $\sum_y Q_{x,y} = 1$, corresponds to a CPTP map $\kappa(\sigma) = \sum_{x,y} Q_{x,y} |y\rangle\langle x| \sigma |x\rangle\langle y|$.
  For example, if the state $\sigma$ corresponds to the random variable $X$, i.e., $\sigma = \sum_{x\in[0:d-1]} p_x |x\rangle \langle x|$, the resultant state after applying $\kappa$ is $\sum_y (\sum_x p_xQ_{x,y}) |y\rangle\langle y|$, i.e., the random variable after applying $Q$ on $X$.
  Sampling a random variable $X$ with the outcome $x$ corresponds to performing PVM $\mathbf{M} = \{ P_{x} = |x\rangle \langle x| \}$ and obtaining the measurement outcome $x$ with probability $p_x$.

\subsection{Stabilizer formalism}
\label{sec:stabform}

Stabilizer formalism is an algebraic structure in quantum information theory and
is often used for the quantum error correction \cite{Gottesman97,Ketkar06}.
In the context of QPIR, it is also an essential tool to design most of the existing multi-server QPIR schemes \cite{song2019capacity,song2019allbutone,song2020colluding,allaix2020quantum}.
With the stabilizer formalism, we will define a new class of QPIR, called stabilizer QPIR in Section~\ref{subsec:defQPIR}, and design our capacity-achieving schemes in Section~\ref{sec:achieve}.
As a preliminary, in this section, we first define stabilizer formalism over finite fields $\Fq$. 
Then, 
    to help understanding how the mathematical definition of the stabilizer formalism is used for information processing tasks, 
    we briefly explain the application to the quantum error correction.

\subsubsection{Stabilizer formalism over finite fields}
Let $q = p^r$ with a prime number $p$ and a positive integer $r$.
Let $\cH$ be a $q$-dimensional Hilbert space spanned by orthonormal states $\{ |j\rangle \mid  j\in \mathbb{F}_q \}$.
For $x\in\mathbb{F}_q$, we define $\bT_x$ on $\mathbb{F}_p^{r}$ as the linear map $y\in\mathbb{F}_q \mapsto xy \in\mathbb{F}_q$ by identifying the finite field $\mathbb{F}_q$ with the vector space $\mathbb{F}_p^{r}$.
Let $\ttr x \coloneqq \trace \bT_x \in\mathbb{F}_p$ for $x\in\mathbb{F}_q$.
Let $\omega \coloneqq \exp({2\pi i/p})$.
For $a,b\in\mathbb{F}_q$, we define unitary matrices
$\mathsf{X}(a) \coloneqq \sum_{j\in\mathbb{F}_q} |j+a\rangle \langle j |$ and $\mathsf{Z}(b) \coloneqq \sum_{j\in\mathbb{F}_q} \omega^{\ttr bj} |j\rangle \langle j |$ on $\cH$.
For $\mathbf{s} = (s_1,\ldots, s_{2\sn}) \in \mathbb{F}_q^{2\sn}$, we define a unitary matrix
$\mathbf{\tilde{W}(s)} \coloneqq 
	\mathsf{X}(s_1) \mathsf{Z}(s_{\sn+1}) \otimes 
	\cdots \otimes \mathsf{X}(s_{\sn}) \mathsf{Z}(s_{2\sn})$ on $\cH^{\otimes \sn}$.
For $\mathbf{x}=(x_1,\ldots,x_\sn),\ \mathbf{y}=(y_1,\ldots,y_\sn) \in \mathbb{F}_q^{\sn}$,
	we define the tracial bilinear form $\langle \mathbf{x}, \mathbf{y} \rangle \coloneqq \ttr \sum_{i=1}^{\sn} x_iy_i\in\mathbb{F}_p$
	and 
	the trace-symplectic bilinear form 
	$\symp{\mathbf{x}}{\mathbf{y}} \coloneqq
	\langle \mathbf{x}, \bJ\mathbf{y} \rangle
	$, where $\bJ$ is a $2\sn \times 2\sn$ matrix
\begin{align*}
\bJ = \begin{pmatrix}
    \mathbf{0} & -\bI_\sn \\ \bI_\sn & \mathbf{0}
    \end{pmatrix}.
\end{align*}
The Heisenberg-Weyl group is defined as
$\mathrm{HW}_q^\sn \coloneqq \px*{c \mathbf{\tilde{W}(s)} \mid  \mathbf{s} \in \mathbb{F}_q^{2\sn},\  c \in \mathbb{C}  }$.
A commutative subgroup of $\mathrm{HW}_q^\sn$ not containing $c \bI_{q^\sn}$ for any $c\neq 0$ is called a {\em stabilizer}.
A subspace $\cV$ of $\mathbb{F}_q^{2\sn}$ is called {\em self-orthogonal} with respect to the bilinear form $\symp{\cdot}{\cdot}$ if
$\cV\subset \cV^{\perp_\mathbb{S}} \coloneqq \{ \mathbf{s}\in \mathbb{F}_q^{2\sn} \mid  \symp{\mathbf{v}}{\mathbf{s}} = 0 \text{ for any } \mathbf{v}\in \cV \}.$ 
Any self-orthogonal subspace of $\mathbb{F}_q^{2\sn}$
defines a stabilizer by the following proposition.

\begin{prop}[{\cite[Section IV-A]{song2020colluding}}] \label{prop:stab}
Let $\cV$ be a self-orthogonal subspace of $\mathbb{F}_q^{2\sn}$.
There exists $\{c_{\mathbf{v}}\in \C \mid \mathbf{v} \in \cV\}$ such that 
\begin{align}
\cS(\cV)  \coloneqq \{ \mathbf{W(v)}  \coloneqq c_{\mathbf{v}} \mathbf{\tilde{W}(v)} \mid  \mathbf{v} \in \cV \} \subset \mathrm{HW}_q^\sn
	\label{eq:123stab}
\end{align}
	is a stabilizer. 
\end{prop}

In the next proposition, we denote the elements of the quotient space $\mathbb{F}_q^{2\sn} / \cV^{\perp_\mathbb{S}}$ by $\overline{\mathbf{s}} \coloneqq \mathbf{s} + \cV^{\perp_\mathbb{S}} \in \mathbb{F}_q^{2\sn} / \cV^{\perp_\mathbb{S}}$.

\begin{prop}[{\cite[Section IV-A]{song2020colluding}}]\label{prop:2stab}
Let $\cV$ be a $d$-dimensional self-orthogonal subspace of $\mathbb{F}_q^{2\sn}$ and $\cS(\cV)$ be a stabilizer defined from Proposition~\ref{prop:stab}.
Then, we obtain the following statements.

\begin{enumerate}[label=(\alph*)]
\item For any $\mathbf{v}\in\cV$, the operation $\mathbf{W(v)} \in \cS(\cV)$ is simultaneously and uniquely decomposed as 
\begin{align}
\mathbf{W(v)} = \sum_{\overline{\mathbf{s}}\in\mathbb{F}_q^{2\sn} / \cV^{\perp_\mathbb{S}} } \omega^{ \symp{\mathbf{v}}{\mathbf{s}} } \bP^{\cV}_{\mathbf{\overline{s}}}
	\label{eq:edededefcm}
\end{align}
with orthogonal projections $\{\bP^{\cV}_{\overline{\mathbf{s}}}\}$ such that
	\begin{align}
	\bP^{\cV}_{\overline{\mathbf{s}}} \bP^{\cV}_{\overline{\mathbf{t}}} &= \mathbf{0} \text{ for any } \overline{\mathbf{s}}\neq \overline{\mathbf{t}}, \\
	\sum_{\overline{\mathbf{s}}\in\mathbb{F}_q^{2\sn} / \cV^{\perp_\mathbb{S}} } \bP^{\cV}_{\overline{\mathbf{s}}} &= \bI_{q^\sn}.
	\label{eq:svcxvijlkfdecomp}
	\end{align} 

\item Let $\cH^{\cV}_{\overline{\mathbf{s}}} \coloneqq \imag{\bP^{\cV}_{\overline{\mathbf{s}}}}$.
We have $\dim \cH^{\cV}_{\overline{\mathbf{s}}} = q^{\sn-d}$ for any $\overline{\mathbf{s}}\in\mathbb{F}_q^{2\sn} / \cV^{\perp_\mathbb{S}}$ and
the quantum system $\cH^{\otimes \sn}$ is decomposed as
\begin{align}
\cH^{\otimes \sn} = \bigotimes_{\overline{\mathbf{s}}\in \mathbb{F}_q^{2\sn} / \cV^{\perp_\mathbb{S}}} \cH^{\cV}_{\overline{\mathbf{s}}} = \mathcal{W} \otimes \mathbb{C}^{q^{\sn-d}},
\label{eq:decompose}
\end{align}
where the system $\mathcal{W}$ is the $q^d$-dimensional Hilbert space spanned by $\{ |\overline{\mathbf{s}}\rangle  \mid \overline{\mathbf{s}}\in \mathbb{F}_q^{2\sn} / \cV^{\perp_\mathbb{S}}  \}$ with the property 
$\cH^{\cV}_{\overline{\mathbf{s}}} = |\overline{\mathbf{s}} \rangle \otimes  \mathbb{C}^{q^{\sn-d}} \coloneqq \{ |\overline{\mathbf{s}} \rangle \otimes | \psi \rangle \mid |\psi\rangle \in\mathbb{C}^{q^{\sn-d}}  \}$.

\item For any $\mathbf{s}, \mathbf{t}\in\mathbb{F}_q^{2\sn}$, we have
	\begin{align}
	\mathbf{W(t)} |\overline{\mathbf{s}} \rangle \otimes  \mathbb{C}^{q^{\sn-d}} 
	&= |\overline{\mathbf{s}+\mathbf{t}} \rangle \otimes  \mathbb{C}^{q^{\sn-d}},
	\label{eq:sdfe123adec}
	\\
	\mathbf{W(t)} \paren*{|\overline{\mathbf{s}} \rangle\langle \overline{\mathbf{s}} | \otimes  \bI_{q^{\sn-d}}  } \mathbf{W(t)}^{\dagger}
	&= |\overline{\mathbf{s}+\mathbf{t}} \rangle \langle \overline{\mathbf{s}+\mathbf{t}} | \otimes  \bI_{q^{\sn-d}}.
	\end{align}
\item For any $\mathbf{v}\in\cV$ and any $|\psi\rangle \in |\overline{\mathbf{0}} \rangle \otimes  \mathbb{C}^{q^{\sn-d}}$, we have
    \begin{align}
        \mathbf{W(v)}  |\psi\rangle
	&= |\psi\rangle.
    \end{align}
\end{enumerate}
\end{prop}

\subsubsection{Application to quantum error correction}
Next, we explain how the stabilizer formalism is used for quantum error correction \cite{Gottesman97, Ketkar06}.
Similar to the classical case, 
 the structure of error correction will be used for accomplishing PIR tasks in the later sections.

Consider the transmission of a quantum state from a sender to a receiver over a noisy channel.
When the sender's message state is $\sigma$ on $\mathbb{C}^{q^{\sn-d}}$, the sender encodes the state $\sigma$ as
$|\overline{\mathbf{0}}\rangle \langle \overline{\mathbf{0}}|\otimes\sigma$ on the quantum system
$|\overline{\mathbf{0}}\rangle \otimes \mathbb{C}^{q^{\sn-d}}\subset \cH^{\otimes \sn}$ defined in (b) of Proposition~\ref{prop:2stab},
 and send the quantum system $\cH^{\otimes \sn}$ to the receiver.
Suppose the noise of the channel is $\mathbf{W(s)}$, \ie the operation $\mathbf{W(s)}$ is applied to the state.
Then, the receiver's state is in the space $|\overline{\mathbf{s}}\rangle \otimes \mathbb{C}^{q^{\sn-d}}$ by (c) of Proposition~\ref{prop:2stab}.
For the decoding of the error,
the receiver detects $\overline{\mathbf{s}}$ by performing the PVM measurement $\{\bP^{\cV}_{\overline{\mathbf{s}}} \mid \overline{\mathbf{s}}\in\Fq^{2\sn}/\cV\}$,
defined from the projections in (a) of Proposition~\ref{prop:2stab}.
This PVM is called {\em syndrome measurement} in the similar context to the classical error correction.
Then, the receiver applies error correction by choosing an element $\mathbf{s'} \in \overline{\mathbf{s}}$ and applying $\mathbf{W(-s')}$, which maps the received state to the original space $|\overline{\mathbf{0}}\rangle\otimes \mathbb{C}^{q^{\sn-d}}$.
Since the noise and error correction operation are combined as the unitray matrix $\mathbf{W(s)}\mathbf{W(-s')}=\mathbf{W(s-s')}$, if $\mathbf{s-s'} \in \cV$, the decoded state is $|\overline{\mathbf{0}}\rangle \langle \overline{\mathbf{0}}|\otimes\sigma$ from (d) of Proposition~\ref{prop:2stab}.
That is, $\sigma$ is correctly recovered by the receiver.
The characterization of the noise $\mathbf{s}$ and the corresponding choice of $\mathbf{s'}$ in decoding are essential problem in quantum error correction to achieve more reliable communication.

\section{Notions of PIR} \label{sec:PIR_def}

\subsection{Classical PIR}

We formally define a classical PIR scheme with MDS-coded storage (MDS-PIR). 
In a general MDS-PIR scheme, 
    one user and $\sn$ servers participate.

\begin{description}
\item[Distributed Storage] The $\sm$ files are given as uniformly and independently distributed random variables $X^1,\ldots, X^{\sm}$ in $\mathbb{F}_q^{\beta\times \sk}$. 
As described in Section~\ref{subsec:storage},
    the files $X = ( (X^1)^\top,\ldots, (X^{\sm})^\top)^\top $ are encoded with an MDS code $\cC$ as $Y = (Y_1,\ldots, Y_{\sn}) = X \bG_\cC \in\mathbb{F}_q^{\beta\sm\times \sn}$
     and 
    is distributed as the $s$-th server contains $Y_s \in \mathbb{F}_q^{\beta\sm} $.
    
\item[Shared Randomness] The servers possibly share randomness $H = (H_1,\ldots, H_\sn)$, where $\cH_s$ is owned by server $s$.

\item[Query]
Let $\tK$ be a uniform random variable with values in $[\sm]$.
The user desiring the $\tK$-th file $X^{\tK}$ prepares $Q^{\tK} = (Q_1^{\tK}, \ldots, Q_\sn^{\tK})$ with local randomness $R$ by the encoder $\mathsf{Enc}_{\mathrm{user}}: [\sm] \times \cR \to \cQ\coloneqq \cQ_1\times \cdots \times \cQ_{\sn}$, where
    $\cR$ is the alphabet of the user's local randomness and $\cQ_s$ is the alphabet of the query to server $s$,
    and sends $Q_s^\tK$ to server $s$.

\item[Response] With the encoder $\mathsf{Enc}_{\mathrm{serv}_s}: \Fq^{\beta\sm} \times \cH_s \times \cQ_s \to \cB_s$, the $s$-th server responds $B_s^\tK = \mathsf{Enc}_{\mathrm{serv}_s}(Y_s,H_s,Q_s) \in\cB_s$ to the user. 
We denote $B^\tK = (B_1^\tK,\ldots, B_\sn^\tK)$ and $\cB = \cB_1\times \cdots \times \cB_\sn$.

\item[Decoding] With the decoder $\mathsf{Dec}:[\sm]\times\cQ\times \cB \to \Fq^{\beta\times \sk}$, the user obtains an estimate $\hat{X}^{\tK} =\mathsf{Dec}(K,Q^\tK,B^\tK) \in\Fq^{\beta\times \sk}$ of $X^{\tK}$.

\end{description}

As described above, an MDS-PIR scheme $\Phi$ is defined as $\Phi_C = (\cC, \sigma_{\mathrm{init}}, \mathsf{Enc}_{\mathrm{user}}, \mathsf{Enc}_{\mathrm{serv}},
\mathsf{Dec})$ with the MDS code for storage $\cC$, the initial state $\sigma_{\mathrm{init}}$, the query encoder of the user $\mathsf{Enc}_{\mathrm{user}}$, the answer encoders of the servers $\mathsf{Enc}_{\mathrm{serv}} \coloneqq \{\mathsf{Enc}_{\mathrm{serv}_s} \mid \forall s\in[\sn]\}$, and the decoder of the user $\mathsf{Dec}$.

The correctness of MDS-PIR is defined as follows.

\begin{defi}[Correctness] \label{def:correctness}
  The correctness of a MDS-PIR scheme $\Phi_C$ is evaluated by the error probability 
    \begin{align}
    P_{\mathrm{err}}(\Phi_C) \coloneqq \max_{\tk\in[\sm]} \Pr[X^{\tk}\neq \hat{X}^{\tk}].
    \end{align}
\end{defi}

    We also consider the following secrecy conditions with a positive integer $\st$ with $1\le \st < \sn$.

\begin{defi}[Privacy with $t$-Collusion]\label{def:privacy}
    \emph{User $\st$-secrecy}: Any set of at most $\st$ colluding servers gains no information about the index $\tk$ of the desired file, \ie
	$p_{Q_{\cT}|\tK=\tk} = p_{Q_{\cT}|\tK=\tk'} $ for any $\tk,\tk'\in[\sm]$ and $\cT\subset[\sn]$ with $|\cT| \leq \st$,
	where $p_{Q_{\cT}|\tK=\tk}$ is the distribution of $Q_{\cT}$ conditioned with $\tK = \tk$.
	
    \emph{Server secrecy}: The user does not gain any information about the files other than the requested one, \ie 
    \begin{align}
        I(B^{\iota} ; X| Q^{\tk}, \tK=\tk) = H(X^{\tk}). \label{eq:sevsec_def}
    \end{align}
\end{defi}

    As customary, we assume that the size of the query alphabet
    is negligible compared to the size of the files. This is well justified if the files are assumed to be large, as the upload cost is independent of the size of the files. For simplicity, we only consider files of sizes $ \sk \b \log q$ in the following.
    However, note that repeatedly applying the scheme with the same queries allows for the download of files that are any multiple of $\sk \b \log q$ in size at the same rate and without additional upload cost.

    When user $\st$-secrecy is satisfied, the scheme is called $[\sn,\sk,\st]$-PIR and leaks no information of the index $\tK$ to any $\st$ colluding servers.
    When both user $\st$-secrecy and server secrecy are satisfied,
        the scheme is called {\em symmetric} and we denote it by $[\sn,\sk,\st]$-SPIR.
        
As a measure of efficiency of the MDS-PIR scheme $\Phi_C$ is defined as follows.
\begin{defi}[MDS-PIR rate] \label{def:MDS-PIR-ratedefi}
The MDS-PIR scheme $\Phi_C$ is defined as
    \begin{align}
	R (\Phi_C)
	= \frac{H(X^i)}{\sum_{j=1}^{\sn} \log |\cB|}.
	\label{eq:ratedef}
	\end{align}
\end{defi}

\begin{defi}[Achievable MDS-PIR rate]
A rate $R$ is called {\em $\epsilon$-error 
    achievable
    $[\sn,\sk,\st]$-PIR ($[\sn,\sk,\st]$-SPIR)
    rate
    with $\sm$ files} if 
    there exists a sequence of 
        $[\sn,\sk,\st]$-PIR ($[\sn,\sk,\st]$-SPIR)
        schemes with $\sm$ files
            $\{\Phi_\ell\}_{\ell}$ such that 
        the PIR rate $R(\Phi_\ell)$ approaches $R$
        and
        the error probability satisfies $\lim_{\ell \to \infty} P_{\mathrm{err}}(\Phi_\ell) \leq \epsilon$.
\end{defi}

\begin{defi}[MDS-PIR capacity] \label{defi:capacity}
The {\em $\epsilon$-error 
    $[\sn,\sk,\st]$-PIR ($[\sn,\sk,\st]$-SPIR)
    capacity with $\sm$ files} 
    $C_{\sm,\epsilon,\mathrm{cl}}^{[\sn,\sk,\st]}$ ($C_{\sm,\epsilon,\mathrm{cl}}^{[\sn,\sk,\st], \mathrm{s}}$) is the supremum of $\epsilon$-error achievable 
    $[\sn,\sk,\st]$-PIR ($[\sn,\sk,\st]$-SPIR)
        rate with $\sm$ files.
\end{defi}

\begin{remark}
Our definition of the achievable rate and capacity with asymptotic $\epsilon$ error generalizes the case of $\epsilon = 0$, which have been discussed in other PIR studies \cite{sun2017replicated,sun2018capacity}.
\end{remark}

We define two well-known classes of classical PIR.
For a set ${\cal I} \subseteq [\sn]$ and $\gamma \in \mathbb{N}$, we define
$\psi_\gamma({\cal I}) \coloneqq \bigcup_{i\in {\cal I}} [(i-1)\gamma + 1 : i\gamma]$.
For example, if ${\cal I} = [\sn]$, we have $\psi_{\gamma}([\sn]) = [\gamma\sn]$.

\begin{defi}[{Linear PIR \cite[Definition 1]{holzbaur2019capacity}}] \label{defi:linearPIR}
A PIR scheme is called {\em linear} if
\begin{itemize}
    \item the query $Q$ is represented by a matrix $\bQ \in\Fq^{\beta \sm \times \gamma \sn}$, where 
        $\bQ_{\psi_\gamma(s)}$ is the query to server $s$, and
    \item the classical answer $B_s$ of server $s$ is represented by 
        \begin{align}
            \bB_{\psi_\gamma(s)} = 
             \bY_{s}^{\top} \bQ_{\psi_\gamma(s)} \in \Fq^{1\times \gamma}.
        \end{align}
\end{itemize}
\end{defi}

We also define strongly linear PIR, which requires the linearity also for the reconstruction of the targeted file.

\begin{defi}[{Strongly linear PIR \cite{holzbaur2019capacity}}] \label{def:stronglyLinear}
A linear PIR scheme is called {\em strongly linear} if
    there exist linear maps $\{f_{i,j} \mid (i,j) \in [\beta] \times [\sk] \}$ such that
    \begin{align*}
    \bX^i_j= f_{i,j} \Big( ( \bB_{(s-1)\gamma+t_{i,j}} \mid s\in[\sn]) \Big)
    \quad \text{for some $t_{i,j}\in[\gamma]$}.
    \end{align*}
\end{defi}

One of our main results is on the MDS-QPIR capacity induced from strongly linear PIR.
The capacity of strongly linear PIR is derived in \cite{holzbaur2019capacity} as follows.

\begin{prop}[{\cite{holzbaur2019capacity, Holzbaur2019ITW}}]
\label{prop:strongly_capacity}
The zero-error capacity of any strongly linear PIR with $[\sn,\sk]$-MDS coded storage and $\st$ colluding servers is 
\begin{align}
	\sup \frac{\sk\beta\log q}{\sum_{i=1}^{\sn} H(B_i)}
	&= 1 - \frac{\sk+\st-1}{\sn}
\end{align}
for any number of files $\sm$.
\end{prop}

\subsection{Quantum PIR (QPIR)}

\subsubsection{QPIR from MDS-coded storage} \label{subsec:defQPIR}

We formally define a QPIR scheme with MDS-coded storage (MDS-QPIR), depicted in Figure~\ref{fig:scheme_model}.

\begin{description}
\item[Distributed Storage] The same as classical PIR.

\item[Shared Entanglement] 	The initial state of the $\sn$ servers is given as a density matrix $\sigma_{\mathrm{init}}$ on quantum system $\cH = \cH_1\otimes \cdots \otimes \cH_{\sn}$, where $\cH_s$ is distributed to server $s$. The state $\sigma_{\mathrm{init}}$ is possibly entangled.

\item[Query] The same as classical PIR.

\item[Response] Each server $s$ applies a CPTP map $\mathsf{Enc}_{\mathrm{serv}_s}[Q_s^\tK,Y_s]$ from $\cH_s$ to $\cA_s$ depending on $Q_s^\tK$ and $Y_s$, 
            where $\cA_s$ is a $\sd$-dimensional quantum system,
        and returns $\cA_s$ to the user.

\item[Decoding] Depending on $K$ and $Q^\tK$, the user applies a POVM $\mathsf{Dec}[K,Q^\tK]$ on $\cA = \cA_1\otimes \cdots \otimes \cA_{\sn}$
        and obtains the measurement outcome $\hat{X}^{\tK}$.

\end{description}

As described above, an MDS-QPIR scheme $\Phi$ is defined as $\Phi = (\cC, \sigma_{\mathrm{init}}, \mathsf{Enc}_{\mathrm{user}}, \mathsf{Enc}_{\mathrm{serv}},
\mathsf{Dec})$ with the MDS code for storage $\cC$, the initial state $\sigma_{\mathrm{init}}$, the query encoder of the user $\mathsf{Enc}_{\mathrm{user}}$, the answer encoders of the servers $\mathsf{Enc}_{\mathrm{serv}} \coloneqq \{\mathsf{Enc}_{\mathrm{serv}_s} \mid \forall s\in[\sn]\}$, and the decoding measurement of the user $\mathsf{Dec}$.

\begin{defi}
  The correctness, privacy, rate, and capacity of QPIR are defined in the same way as Definitions~\ref{def:correctness}, \ref{def:privacy}, \ref{def:MDS-PIR-ratedefi}, and \ref{defi:capacity}, respectively, except that \eqref{eq:sevsec_def} and \eqref{eq:ratedef} are replaced as 
    \begin{align}
    I(\cA ; X| Q^{\tk}, \tK=\tk) = H(X^{\tk}),
    \end{align}
    and 
    \begin{align}
	R (\Phi)
	= \frac{H(X^i)}{\sum_{j=1}^{\sn} \log \dim \cA}.
    \end{align}
\end{defi}

\begin{nota}
We denote by 
    $C_{\sm,\epsilon}^{[\sn,\sk,\st]}$ ($C_{\sm,\epsilon}^{[\sn,\sk,\st], \mathrm{s}}$) the {\em $\epsilon$-error 
    $[\sn,\sk,\st]$-QPIR ($[\sn,\sk,\st]$-QSPIR)
    capacity with $\sm$ files}.
\end{nota}

In Definition~\ref{def:privacy}, user $\st$-secrecy is defined as the independence of the index $K$ and the queries $Q_{\cal T}$ of the colluding servers.
Although this user secrecy condition is natural in classical PIR, one may be unsure whether this condition is sufficient for the QPIR setting because the servers share quantum entanglement.
To justify this condition in the QPIR setting, we consider the malicious scenario where the servers apply malicious operations on the answered systems in order to extract the information of the user's request $K$.
Even in this malicious scenario, 
the servers cannot 
    exploit entanglement to break the user's secrecy because of the no-signaling principle \cite{Bell}.
No-signaling principle states that two parties sharing an entangled state cannot communicate any information from their local measurements.
From this principle, even if the colluding servers share entanglement with the other servers or the user throughout the scheme, the only information obtained by the colluding servers is the queries $Q_{\cal T}$.
Thus, the user $\st$-secrecy condition guarantees the secrecy of $K$ from the colluding servers.

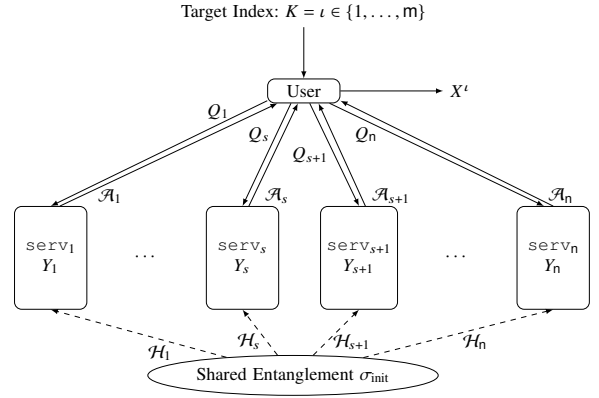
\begin{figure}[t]
\begin{center}
\resizebox{0.85\linewidth}{!}{
\begin{tikzpicture}[node distance = 3.3cm, every text node part/.style={align=center}, auto]
    \node [block] (user) {User};
    \node [above=1cm of user] (sent) {Target Index: $\tK = \tk \in\{1,\ldots,\sm\}$};
    \node [block,minimum height = 2cm, below left=2cm and 3.5cm of user] (serv1) {$\mathtt{serv}_1$\\$Y_1$};
    \node [right=0.8cm of serv1] (ten) {$\cdots$};
    \node [block,minimum height = 2cm, right=0.8cm of ten] (serv2) {$\mathtt{serv}_{s}$\\$Y_s$};
    \node [block,minimum height = 2cm, right=0.8cm of serv2] (serv3) {$\mathtt{serv}_{s+1}$\\$Y_{s+1}$};
    \node [right=0.8cm of serv3] (ten2) {$\cdots$};
    \node [block,minimum height = 2cm, right=0.8cm of ten2] (servn) {$\mathtt{serv}_{\sn}$\\$Y_{\sn}$};
    
    \node [draw,ellipse,below right=1cm and 2cm of serv1] (shared) {Shared Entanglement $\sigma_{\mathrm{init}}$};
    
    \node [right=2cm of user] (receiv) {$X^\tk$};
    
    \path [line] (sent) -- (user);
    
    \path [line] (user.195) --node[pos=0.1,left=2mm] {$Q_1$} (serv1.north);
    \path [line] (user.224) --node[pos=0.3,left] {$Q_{s}$} (serv2.north);
    \path [line] (user) --node[pos=0.5,left] {$Q_{s+1}$} (serv3.north);
    \path [line] (user) --node[pos=0.3,left=2mm] {$Q_{\sn}$} (servn.100);

    \path [line] (serv1.80)--node[pos=0.1,right=2mm] {$\cA_1$} (user);
    \path [line] (serv2.83) --node[pos=0.1,right=1mm] {$\cA_{s}$} (user);
    \path [line] (serv3.83) --node[pos=0.1,right=1mm] {$\cA_{s+1}$} (user.320);
    \path [line] (servn.north) --node[pos=0.1,right=2mm] {$\cA_{\sn}$} (user.345);
    
    \path [line] (user.east) -- (receiv);
    
    \path [line,dashed] (shared) --node[pos=0.1,left=6mm] {$\cH_{1}$} (serv1.south);
    \path [line,dashed] (shared) --node[pos=0.3,left=0mm] {$\cH_{s}$} (serv2.south);
    \path [line,dashed] (shared) --node[pos=0.3,right=0mm] {$\cH_{s+1}$} (serv3.south);
    \path [line,dashed] (shared) --node[pos=0.33,right=6mm] {$\cH_{\sn}$} (servn.south);
\end{tikzpicture}
}
\caption{Quantum private information retrieval scheme.}\label{fig:scheme_model}
\end{center}
\end{figure}

\subsubsection{Example of QPIR scheme} \label{sec:example_stab}

With stabilizer formalism, we give an example of two-server QPIR, which corresponds to the QPIR scheme in \cite{song2019capacity}.
Let $\cH_1$ and $\cH_2$ be two-dimensional quantum systems, which are also called {\em qubits}.
From Proposition~\ref{prop:stab}, we define a stabilizer on $\cH_1\otimes \cH_2$ with the self-orthogonal subspace
\begin{align*}
\cV = \{ (0,0,0,0), (1,1,0,0), (0,0,1,1), (1,1,1,1) \} \subset \mathbb{F}_2^{4}.
\end{align*}
The space $\cV$ satisfies $\cV = \cV^{\perp_\mathbb{S}}$.
With this stabilzer, 
we set the initial entangled state of the two servers as $|\overline{\mathbf{0}}\rangle \in \cH_1\otimes \cH_2$, where $\overline{\mathbf{s}} = \mathbf{s} + \cV^{\perp_\mathbb{S}}$ for all $\mathbf{s}\in\mathbb{F}_2^4$.
The two servers have $\cH_1$ and $\cH_2$, respectively.
The files are prepared as $\mathbf{m}_i = (m_{iX},m_{iZ}) \in \mathbb{F}_2^2$ for all $i\in[\sm]$.
For querying the $k$-th file, the user sends queries
\begin{align}
    \mathbf{q}_1 &= (\mathbf{e}_k,\mathbf{e}_k) + \mathbf{r}\in\mathbb{F}_2^{2\sm},\\
    \mathbf{q}_2 &= \mathbf{r}\in\mathbb{F}_2^{2\sm},
\end{align}
where
    $\mathbf{e}_k$ is the $k$-th standard vector in $\mathbb{F}_2^{\sm}$
    and
    $\mathbf{r}$ is a random vector in $\mathbb{F}_2^{2\sm}$.
After receiving queries, the servers generates 
    \begin{align}
        (a_1,b_1) &= \mathbf{q}_1 \cdot \bm = \bm_k + \br \cdot \bm,\\
        (a_2,b_2) &= \mathbf{q}_2 \cdot \bm = \br \cdot \bm,
    \end{align}
where $\bm = ((m_{1X},m_{2X},\ldots, m_{\sm X}), (m_{1Z},m_{2Z},\ldots, m_{\sm Z}))\in\mathbb{F}_2^{2\sm}$.
Then, the server $i$ applies $\sX(a_i)\sZ(b_i)$ on $\cH_i$ and sends $\cH_i$ to the user.
The user receives the states 
\begin{align}
    \mathbf{\tilde{W}}(a_1,a_2,b_1,b_2) |\overline{\mathbf{0}} \rangle 
    &= |\overline{(a_1,a_2,b_1,b_2)} \rangle\\
    &= |\overline{(m_{kX},0,m_{kZ},0)} \rangle,
    \label{eq:sexample_state}
\end{align}
where the first equality follows from \eqref{eq:sdfe123adec}
    and
    the second equality follows from $(1,1,1,1)\in\cV^{\perp_\mathbb{S}} =\cV$.
By applying measurement on the received state,
    the user retrieves $\bm_k = (m_{kX},m_{kZ})\in\mathbb{F}_2^2$ correctly.
The user secrecy is satisfied from the query structure, and the server secrecy is satisfied because the user's state only depends on $\bm_k$ as in \eqref{eq:sexample_state}.
The QPIR rate is $1$ because $2$ bits are retrieved and $2$ qubits are downloaded.

\subsubsection{Classes of QPIR}

As a general class of QPIR schemes, we introduce a new class called stabilizer QPIR, which includes 
the example in Section~\ref{sec:example_stab} and most of the known multi-server QPIR schemes \cite{song2019capacity, song2020colluding, song2019allbutone, allaix2020quantum}.

\begin{defi}[Stabilizer QPIR]
A QPIR scheme is called a stabilizer QPIR induced from a classical PIR scheme $\Phi_C$ if
\begin{itemize}
\item the initial state of the servers $\sigma_{\mathrm{init}}$ is a state in $\cH^{\cV}_{\overline{\mathbf{0}}} = |\overline{\mathbf{0}}\rangle \otimes \mathbb{C}^{q^{\sn-d}} \subset \cH^{\otimes \sn}$ defined with a self-orthogonal subspace $\cV$ by Proposition~\ref{prop:2stab},
\item the query is the same as $\Phi_C$, and
\item the $s$-th server's operation is the Weyl operation $\sX(a_s)\sZ(b_s)$, where $(a_s,b_s               )\in\Fq^2$ is the $s$-th server's answer of $\Phi_C$.
\end{itemize}
\end{defi}

In Section~\ref{sec:achieve}, we construct a stabilizer QPIR scheme, which achieves the capacities in Corollaries \ref{coro:TQPIR} and \ref{coro:QPIR}.

Further, we define a more general class of QPIR as follows.

\begin{defi}[Dimension-squared QPIR]
A QPIR scheme is said to be {\em dimension-squared} if the $s$-th server's operation is determined by classical information $B_s\in   \cB_s$ with $|\cB_s|\le \sd^2$ for all $s\in[\sn]$.

Furthermore, if $B = (B_1,\ldots, B_{\sn})$ is the answer of a classical PIR scheme $\Phi_C$ and the query of the QPIR scheme is the same as $\Phi_C$, the QPIR scheme is called a {\em dimension-squared QPIR induced from the classical PIR scheme $\Phi_C$}.
\end{defi}

Any stabilizer QPIR scheme is a dimension-squared scheme induced from a classical PIR scheme.
Accordingly, the example in Section~\ref{sec:example_stab} and the multi-server QPIR schemes \cite{song2019capacity, song2020colluding, song2019allbutone, allaix2020quantum} are also dimension-squared schemes induced from strongly linear schemes.
In Section~\ref{sec:converse}, we derive the converse bound 
for dimension-squared QPIR schemes.

When a classical PIR scheme $\Phi_C$ induces a QPIR scheme without the condition of dimensions, then the scheme can be modified to induce dimension-squared QPIR in the following way.
First, we make the $\sn$ answers the same size by repeating $\Phi_C$ multiple times while reordering the roles of the servers for all possible cases.
Let $\sd'$ be the size of one answer 
    and $\Phi_C'$ be the repeated PIR scheme.
Again, let $\Phi_C''$ be the PIR scheme made by repeating $\Phi_C'$ $\sd'$ times,
    and then,
    the size of each answer of $\Phi_C''$ is $(\sd')^2$.
Thus, a dimension-squared QPIR scheme $\Phi_Q$ is induced from $\Phi_C''$ if $\Phi_Q$ can be made to satisfy the correctness condition.
For convenience, we consider a dimension-squared QPIR scheme induced from $\Phi_C''$ as induced from $\Phi_C$.

\begin{nota}
We denote by 
    $C_{\sm,\epsilon,\mathrm{stab}}^{[\sn,\sk,\st]}$,
    $C_{\sm,\epsilon,\mathrm{dim}}^{[\sn,\sk,\st]}$
    ($C_{\sm,\epsilon,\mathrm{stab}}^{[\sn,\sk,\st], \mathrm{s}}$, $C_{\sm,\epsilon,\mathrm{dim}}^{[\sn,\sk,\st],\mathrm{s}}$) 
    the $\epsilon$-error 
        $[\sn,\sk,\st]$-QPIR ($[\sn,\sk,\st]$-QSPIR) capacities of 
        stabilizer QPIR induced from strongly linear PIR and dimension-squared QPIR induced from strongly linear PIR.
\end{nota}

From the definitions, the capacities are decreasing for $\st$ and increasing for $\epsilon$, and satisfy
    \begin{alignat}{4}
    C_{\sm,\epsilon,\text{stab}}^{[\sn,\sk,\st],\mathrm{s}}  &\le &\quad C_{\sm,\epsilon,\textnormal{dim}}^{[\sn,\sk,\st],\mathrm{s}}\quad     &\!\le &\quad  C_{\sm,\epsilon}^{[\sn,\sk,\st],\mathrm{s}} \nonumber\\
      \vle \quad\       &\phantom{\supset} &   \vle\qquad            &\!\phantom{\supset} &   \vle \quad    \
  \label{eq:inclusions}\\
       C_{\sm,\epsilon,\text{stab}}^{[\sn,\sk,\st]}  &\le &\quad C_{\sm,\epsilon,\textnormal{dim}}^{[\sn,\sk,\st]}\quad     &\!\le &\quad  C_{\sm,\epsilon}^{[\sn,\sk,\st]} .
       \nonumber
\end{alignat}

\begin{table}[t]   
\begin{center}
\caption{Summary of symbols} \label{tab:1}
\begin{tabular}{|c|c|}
\hline
Symbol & Description \\
\hline
\hline
$\sn$ & Number of servers / Length of a code \\
\hline
$\sk$ & Dimension of $[\sn,\sk]$-MDS code \\
\hline
$\st$ & Number of colluding servers / Dimension of query code \\
\hline
$(i), \sm$ & (Index running over) Number of files \\
\hline
$\sd$ & Dimension of the answer $\cA_s$ ($\forall s\in[\sn]$) \\
\hline
$(b), \beta$ & (Index running over) Number of stripes in a file \\
\hline
$(r), \rho$ & (Index running over) Number of rounds \\
\hline
$p,s$ & Indices of pair and server, respectively \\
\hline
$\cC,\cD,\cS$ & Storage, query and star-product codes \\
\hline
$(\sigma),\cH$ & (State of) Quantum system \\
\hline
$\cV$ & Self-orthogonal subspace of $\Fq^{2\sn}$ \\
\hline
$\sX,\sZ$ & Pauli operators \\
\hline
$\bX,\bY$ & Matrices of files and encoded symbols \\
\hline
$\bZ,\bB$ & Matrices of queries and responses \\
\hline
$\Phi_C,\Phi$ & Classical, quantum PIR scheme  \\
\hline
\end{tabular}
\end{center}
\end{table}

\section{Main results} \label{sec:main}

In this section, we give our two main results of the paper.
The first result is the asymptotic capacity of stabilizer QPIR and dimension-squared QPIR induced from strongly linear PIR.
The second result is the general asymptotic capacity without collusion, i.e., the case $\st = 1$.
Before our capacity result, we state a general upper bound of dimension-squared QPIR capacity.

\begin{theo}[Converse for dimension-squared QPIR  induced from classical PIR] \label{theo:2C}
Let $A$ be a set of assumptions on classical PIR and $C_\epsilon[A]$ be the $\epsilon$-error capacity of the classical PIR with assumptions $A$.
Then, for any $\epsilon'\in[0,1)$, the $\epsilon'$-error capacity of dimension-squared QPIR induced from classical $\epsilon$-error PIR with the assumptions $A$ is upper bounded by $\min\{1,2C_\epsilon[A]\}$.
\end{theo}

Theorem~\ref{theo:2C} will be proved in Section~\ref{subsec:conv2}.
Notice that Theorem~\ref{theo:2C} is proved for dimension-squared QPIR induced from any classical PIR class.
Intuitively, the dimensional condition in the dimension-squared QPIR is the key factor for doubling the capacity of any classical PIR.
On the other hand, it should be noted that classical PIR schemes do not necessarily induce QPIR schemes, \ie the existence and the construction of QPIR induced from the classical PIR is not trivial as discussed in Section~\ref{sec:contrib}.

Our first capacity result is on the capacities of stabilizer QPIR and dimension-squared QPIR induced from strongly linear PIR.
An upper bound of the capacities $C_{\sm,\epsilon,\textnormal{dim}}^{[\sn,\sk,\st]}$ and $C_{\sm,\epsilon,\textnormal{dim}}^{[\sn,\sk,\st], \mathrm{s}}$ is derived by Theorem~\ref{theo:2C} and Proposition~\ref{prop:strongly_capacity} as
\begin{align}
C_{\sm,0,\textnormal{dim}}^{[\sn,\sk,\st]}
, C_{\sm,0,\textnormal{dim}}^{[\sn,\sk,\st], \mathrm{s}}
\le 2 \paren*{ 1- \frac{\sk+\st-1}{\sn}}.
\label{eq:convds}
\end{align}
Furthermore, we prove the following theorem in Section~\ref{sec:achieve}. 

\begin{theo}[Achievability] \label{theo:achieve}
Let $\sn,\sk,\st$ be positive integers with $1\le \sn/2 \leq \sk+\st-1 < \sn$.
There exists a stabilizer QPIR scheme induced from strongly linear PIR with $[\sn,\sk]$-MDS coded storage and $\st$-colluding servers achieving \eqref{eq:convds} with equality for any number of files $\sm$ and without error.
\end{theo}
Combining Eqs.~\eqref{eq:inclusions}, \eqref{eq:convds}, and Theorem \ref{theo:achieve}, we obtain the first capacity result.

\begin{coro}[{MDS-Q(S)PIR capacity with colluding servers}] \label{coro:TQPIR}
Let $\sn, \sk, \st$ be positive integers such that $1 \le \sk \le \sn$ and $1 \le \st < \sn$. Then, 
 for any $C_{\sm,0} \in \{C_{\sm,0,\textnormal{stab}}^{[\sn,\sk,\st]} 
    ,  C_{\sm,0,\textnormal{stab}}^{[\sn,\sk,\st], \mathrm{s}} 
    , C_{\sm,0,\textnormal{dim}}^{[\sn,\sk,\st]} 
    , C_{\sm,0,\textnormal{dim}}^{[\sn,\sk,\st], \mathrm{s}}\}$,
   \begin{align}
    C_{\sm,0} =
	\begin{cases}
	1  &   \text{if $\sk+\st-1 \leq \sn/2$},\\
	2 \paren*{ 1- \frac{\sk+\st-1}{\sn}} & \text{otherwise}.
	\end{cases}
    \label{eq:L_capacity}
   \end{align}
\end{coro}

In Corollary~\ref{coro:TQPIR}, the case for $\sk+\st-1 \leq \sn/2$ is proved as follows.
When $\sk+\st-1 = \sn/2$, Theorem~\ref{theo:achieve} proves the rate $1$ is achievable.
If $\st \leq \st'$, the QPIR scheme for $\st'$ colluding servers also has the user secrecy against $\st$ colluding servers.
Therefore, when $\sk+\st-1 \leq \sn/2$,
    we can apply the scheme for $\sk+\st'-1 = \sn/2$ with $\sn$ even to achieve the rate $1$.
Finally, the tightness of the rate $1$ follows trivially from definition. If $\sn$ is odd, we just consider $\sn-1$ servers and $\st=(\sn+1)/2 - \sk$ in order to achieve rate 1.

As the second result, when no servers collude, \ie $\st=1$, we prove the general asymptotic capacity theorem. 
Without the assumption of dimension-squared QPIR, we prove the following upper bound of QPIR.

\begin{theo}[Converse of QPIR without collusion] \label{theo:conv22}
Let $\sn,\sk$ be positive integers with $1\le \sn/2 \leq \sk < \sn$.
Then, we have 
    \begin{align}
    \lim_{\epsilon\to 0} \lim_{\sm \to \infty} C_{\sm,\epsilon}^{[\sn,\sk,1]}
	\leq
	2\paren*{1- \frac{\sk}{\sn}}.
    \label{eq:capacity)cc}
	\end{align}
\end{theo}
Theorem~\ref{theo:conv22} will be proved in Section~\ref{subsec:conv2}.
Combining Eq.~\eqref{eq:inclusions}, Theorem~\ref{theo:conv22}, and Theorem~\ref{theo:achieve} for the case $\st=1$, we obtain the second capacity result.
\begin{coro}[{MDS-Q(S)PIR capacity}] \label{coro:QPIR}
Let $\sn, \sk$ be positive integers such that $1 \le \sk \le \sn$.
 For any $C_{\sm,\epsilon} \in \{C_{\sm,\epsilon,\textnormal{stab}}^{[\sn,\sk,1]} 
    ,  C_{\sm,\epsilon,\textnormal{stab}}^{[\sn,\sk,1], \mathrm{s}} 
    , C_{\sm,\epsilon,\textnormal{dim}}^{[\sn,\sk,1]} 
    , C_{\sm,\epsilon,\textnormal{dim}}^{[\sn,\sk,1], \mathrm{s}}
    , C_{\sm,\epsilon}^{[\sn,\sk,1]}
    , C_{\sm,\epsilon}^{[\sn,\sk,1], \mathrm{s}}\}$,
    \begin{align}
    \lim_{\epsilon\to 0} \lim_{\sm \to \infty} C_{\sm,\epsilon}
	= 
	\begin{cases}
	1  &   \text{if $\sk \leq \sn/2$},\\
	2\paren*{1- \frac{\sk}{\sn}}
        & \text{otherwise}.
    \end{cases}
    \label{eq:capacity}
	\end{align}
\end{coro}
In Corollary~\ref{coro:QPIR},
    the smallest capacity in the six capacities is $C_{\sm,\epsilon,\textnormal{stab}}^{[\sn,\sk,1], \mathrm{s}}$ from \eqref{eq:inclusions}, and this value is asymptotically lower bounded by the RHS of \eqref{eq:capacity} from Theorem~\ref{theo:achieve}.
On the other hand, the greatest capacity is $C_{\sm,\epsilon}^{[\sn,\sk,1]}$, which is upper bounded by the RHS of \eqref{eq:capacity} from Theorem~\ref{theo:conv22}.

\section{Achievability} \label{sec:achieve}

We will frequently deal with $\sm \b \times 2\sn$ matrices, where sub-blocks of $\beta$ rows and the pair of columns $s$ and $\sn+s$ semantically belong together. We therefore index such a matrix $\bY$ by two pairs of indices $(i,b),\ i\in[\sm],\ b\in[\beta]$ and $(p,s),\ p\in [2],\ s\in [\sn]$, where $\bY^{i,b}_{p,s}$ denotes the symbol in row $(i-1)\beta + b$ and column $(p-1)\sn+s$, \ie the symbol in the $b$-th row of the $i$-th sub-block of rows and the $s$-th column of the $p$-th sub-block of columns. Omitting of an index implies that we take all positions, \ie $\bY^{i}$ denotes the $i$-th sub-block of $\b$ rows, $\bY^{i,b}$ the row $(i-1)\beta + b$, $\bY_{p}$ the $p$-th sub-block of $\sn$ columns, and $\bY_{p,s}$ the column $(p-1)\sn+s$. For the reader's convenience, we sometimes imply the separation of the sub-blocks of columns by a vertical bar in the following. We denote by $\be^\lambda_\gamma$ the standard basis column vector of length $\lambda$ in $\Fq^\lambda$ with a 1 in position $\gamma \in [\lambda]$. Given $a \in [\a],\ b \in [\b]$, it will help our notation to call \emph{coordinate} $(a,b)$ the position $\b (a - 1) + b$ in a vector of length $\a\b$. For instance, $\be^{2\cdot 3}_{(2,1)}=\be^6_4=(0,0,0,1,0,0)$. For a zero matrix $\mathbf{0}$ and matrices $\bM_1,\bM_2 \in \Fq^{\mu \times \nu}$
\[
\diag{\bM_1,\bM_2} = \ppmatrix{\bM_1 & \mathbf{0} \\ \mathbf{0} & \bM_2} \in \Fq^{2\mu \times 2\nu}.
\]
For a matrix $\bM$, the space spanned by the rows of $\bM$ is denoted by $\langle \bM\rangle_{\mathsf{row}}$.

For two vectors $\bc,\bd\in \Fq^\sn$ we define the (Hadamard-) star-product as $\bc\star \bd = (c_1 d_1,c_2 d_2,\ldots,c_\sn d_\sn)$. For two codes $\cC, \cD \subseteq \F^{\sn}$ we denote $\cC \star \cD = \langle \{ \bc \star \bd \ | \ \bc\in \cC, \bd\in \cD \} \rangle$.
Observe that, as the star-product is an element-wise operation, we have 
\begin{equation}\label{eq:starCartesian}
    \paren{\cC \times \cC} \star \paren{\cD \times \cD} = \paren{\cC \star \cD} \times \paren{\cC \star \cD} \ .
\end{equation}

\subsection{Generalized Reed--Solomon codes}

We consider systems encoded with (the Cartesian product of) Generalized Reed--Solomon (GRS) codes (cf.~\cite[Ch.~10]{macwilliams1977theory}), a popular class of MDS codes.

\begin{defi}
Let $\cL = \{ \a_i \in \Fq : i \in [n] \}$ and $\cM = \{ \b_i \in \Fq : i \in [n] \}$ be the sets of the \emph{code locators} and of the \emph{column multipliers}, respectively. The \emph{Generalized Reed--Solomon} (GRS) code $\cC$ of dimension $k$ is given by
\[
\cC = \{ (\b_1 f(\a_1),\ldots, \b_n f(\a_n))\ :\ f \in \Fq[x],\deg(f)<k \}.
\]
\end{defi}

Among coded storage systems, these have proven to be particularly well-suited for PIR and general schemes exist for a wide range of parameters \cite{tajeddine2018private, freij2017private, tajeddine2019private}. The key idea is to design the queries such that the retrieved symbols are the sum of a codeword of another GRS code (of higher dimension), which we refer to as the \emph{star-product code}, plus a vector depending only on the desired file. To obtain the desired file, the codeword part is projected to zero, leaving only the desired part of the responses. In the QPIR system we consider in the following, this projection is part of the quantum measurement. This imposes a constraint on this star-product code, namely, that the code is (weakly) self-dual. In the following, we collect/establish the required theoretical results on GRS codes and their star-products.

\begin{defi}[Weakly self-dual code]
  \label{def:WeaklySelfDualCode}
  We say that an $[\sn,\sk]$ code $\cC$ is \emph{weakly self-dual} if $\cC^\perp \subseteq \cC$ and \emph{self-dual} if $\cC^\perp = \cC$. It is easy to see that any such code with parity-check matrix $\bH$ has a generator matrix of the form $\bG = (\bH^\top \ \ \bF^\top)^\top$ for some $(2\sk-\sn)\times \sn$ matrix $\bF$.
\end{defi}

\begin{lemm}[Follows from {\cite[Theorem~3]{grass2008self}}]\label{lem:existenceSelfDual}
  For $q=2^r$ there exist self-dual GRS $[2\sk,\sk]$ codes over $\Fq$ for any $\sk\in [2^{r-1}]$ and code locators $\cL$.
\end{lemm}

\begin{lemm}\label{lem:existenceSelfDualAnyk}
  Let $q$ be even with $q\geq \sn$. Then there exists a weakly self-dual $[\sn,\sk]$ GRS code $\cC$ for any integer $\sk\geq \frac{\sn}{2}$ and code locators~$\cL$.
\end{lemm}
\begin{proof}
    First consider the case of even $n$.
  Let $\cS$ be an $[\sn,\sn/2]$ self-dual GRS code with code locators $\mathcal{L} \subseteq \F_{q}$, as shown to exist in \cite[Theorem~3]{grass2008self} (see Lemma~\ref{lem:existenceSelfDual}). It is easy to see that this code is a subcode of the $[\sn,\sk]$ GRS code $\cC$ with the same locators and column multipliers. The property $\cC^\perp \subset \cC$ follows directly from observing that $\cC^\perp \subseteq \cS^\perp = \cS \subseteq \cC$.
  
  Now consider the case of odd $\sn$. First, observe that this implies $\sn<q$ and $\ceil{\frac{\sn}{2}} = \frac{\sn+1}{2}$. Then, by Lemma~\ref{lem:existenceSelfDual}, there exists a self-dual $[\sn+1,\ceil{\frac{\sn}{2}}]$ GRS code $\cS'$ with code locators $\cL' = \cL \cup \{\alpha\}$, where $\alpha \in \F_q \setminus \cL$. Let $j\in [\sn+1]$ be the index of the position corresponding to $\alpha$. 
  Now consider the code $\cC$ obtained from puncturing this position $j$, \ie the set
  \begin{align*}
      \cS = \{ \bc_{[\sn+1]\setminus \{j\}}  \ | \ \bc \in \cS' \} \ .
  \end{align*}
  It is well-known that the operation dual to puncturing is shortening and therefore the corresponding $[\sn,\ceil{\frac{\sn}{2}}-1]$ dual code $\cS^\perp$ is given by
  \begin{align*}
      \cS^\perp &= \{ \bc_{[\sn+1]\setminus \{j\}}  \ | \ c_j = 0, \bc \in (\cS')^\perp \} \\
      &=  \{ \bc_{[\sn+1]\setminus \{j\}}  \ | \ c_j = 0, \bc \in \cS' \} \ .
  \end{align*}
  Clearly, this operation preserves the weak duality, \ie $ \cS^\perp \subset  \cS$. Again, it is easy to see that $\cS$ is a subcode of the $[\sn,\sk]$ GRS code $\cC$ with the same locators and column multipliers for any $\sk \geq \frac{\sn}{2}$. The statement follows from observing that we have $\cC^\perp \subseteq \cS^\perp \subset \cS \subseteq \cC$
\end{proof}

\begin{lemm}
 Let $q$ be even with $q\geq n$. For any $[\sn,\sk]$ GRS code $\cC$ there exists an $[\sn,\st]$ GRS code $\cD$ such that their star-product $\cS = \cC \star \cD$ is an $[\sn,\sk+\st-1]$ weakly self-dual GRS code.
  \label{lem:SelfDualStarCode}
\end{lemm}
\begin{proof}
  By \cite{mirandola2015critical} the star product between an $[\sn,\sk]$ GRS code $\cC$ with column multipliers $\cM_\cC$ and an $[\sn,\st]$ GRS code $\cD$ with column multipliers $\cM_\cD$, both with the same locators $\cL$, is the $[\sn,\sk+\st-1]$ GRS code with column multipliers $\cM_\cC\star \cM_\cD$ and code locators $\cL$. Denote by $\cM_\cS$ the column multipliers of a weakly-self dual $[\sn,\sk+\st-1]$ GRS code with code locators $\cL$, as shown to exist in Lemma~\ref{lem:existenceSelfDualAnyk}. Then, the lemma statement follows from setting $\cM_\cD = (\cM_\cC)^{-1}\star \cM_\cS$, where we denote by $(\cM_\cC)^{-1}$ the element-wise inverse of $\cM_\cC$.
\end{proof}

\subsection{Description of the coded QPIR scheme}
\label{sec:scheme}

In this subsection we describe the required preliminaries for the capacity-achieving QPIR scheme. Afterwards, we give a compact list of the steps followed by the protocol.

\textbf{Storage.} We consider a linear code $\cC$ of length $2\sn$ and dimension $2\sk$, which is the Cartesian product of an $[\sn,\sk]$ GRS code $\cC'$ over $\Fq$ with itself\footnote{We choose this description of the storage code because this structure is required for the quantum PIR scheme. However, note that the system can equivalently be viewed as being encoded with an $[\sn,\sk]$ code over $\F_{q^2}$, where each of the servers stores one column of the resulting codeword matrix.}, \ie $\cC = \cC' \times \cC'$. It therefore has a generator matrix $\bG_\cC = \diag{\bG_{\cC'},\bG_{\cC'}}$, where $\bG_{\cC'}$ is a generator matrix of $\cC'$. The $\sm\beta \times 2\sn$ matrix of encoded files is given by $\bY = \bX \cdot \bG_\cC$. Server $s \in [\sn]$ stores columns $s$ and $\sn+s$ of $\bY$, \ie it stores $\bY_{1,s}$ and $\bY_{2,s}$ (for an illustration see Figure~\ref{fig:DSS2}). 
For a given integer $c$, 
    which will be defined in the next paragraph,     the parameter $\b$ is fixed to $\b = \lcm(c,\sk)/\sk$.

\begin{figure*}[ht]
\centering
\begin{tikzpicture}[thick,scale=0.9, every node/.style={transform shape}]
\path
(4.35,0) node{
	$\begin{pmatrix}[ccc|ccc]
	    \bX_{1,1}^{1,1}  & \cdots & \bX_{1,\sk}^{1,1} & \bX_{2,1}^{1,1}  & \cdots & \bX_{2,\sk}^{1,1}  \\
	    \vdots & \ddots & \vdots & \vdots & \ddots & \vdots \\
	    \bX_{1,1}^{1,\b} & \cdots & \bX_{1,\sk}^{1,\b} & \bX_{2,1}^{1,\b} & \cdots & \bX_{2,\sk}^{1,\b} \\ \hline
	    \vdots & \vdots & \vdots & \vdots & \vdots & \vdots \\ \hline
	    \bX_{1,1}^{\sm,1}  & \cdots & \bX_{1,\sk}^{\sm,1} & \bX_{2,1}^{\sm,1}  & \cdots & \bX_{2,\sk}^{\sm,1} \\
	    \vdots & \ddots & \vdots & \vdots & \ddots & \vdots \\
	    \bX_{1,1}^{\sm,\b} & \cdots & \bX_{1,\sk}^{\sm,\b} & \bX_{2,1}^{\sm,\b} & \cdots & \bX_{2,\sk}^{\sm,\b}
	\end{pmatrix} \quad \cdot \bG_\cC = \quad
	\begin{pmatrix}[ccc|ccc]
	    \bY_{1,1}^{1,1}  & \cdots & \bY_{1,\sn}^{1,1} & \bY_{2,1}^{1,1}  & \cdots & \bY_{2,\sn}^{1,1}  \\
	    \vdots & \ddots & \vdots & \vdots & \ddots & \vdots \\
	    \bY_{1,1}^{1,\b} & \cdots & \bY_{1,\sn}^{1,\b} & \bY_{2,1}^{1,\b} & \cdots & \bY_{2,\sn}^{1,\b} \\ \hline
	    \vdots & \vdots & \vdots & \vdots & \vdots & \vdots \\ \hline
	    \bY_{1,1}^{\sm,1}  & \cdots & \bY_{1,\sn}^{\sm,1} & \bY_{2,1}^{\sm,1}  & \cdots & \bY_{2,\sn}^{\sm,1} \\
	    \vdots & \ddots & \vdots & \vdots & \ddots & \vdots \\
	    \bY_{1,1}^{\sm,\b} & \cdots & \bY_{1,\sn}^{\sm,\b} & \bY_{2,1}^{\sm,\b} & \cdots & \bY_{2,\sn}^{\sm,\b}
	\end{pmatrix}$
}
(-3.8,1.1) node[blue]{file 1}
(-3.8,-1) node[blue]{file $\sm$}
(6.1,-2.3) node[orange]{\small $\text{\rmfamily\scshape server}\, 1$}
(7.92,-2.33) node[orange]{\small $\text{\rmfamily\scshape server}\, \sn$}
(9.38,-2.3) node[orange]{\small $\text{\rmfamily\scshape server}\, 1$}
(11.2,-2.33) node[orange]{\small $\text{\rmfamily\scshape server}\, \sn$}
;
\draw[thin,blue,rounded corners=4pt] (-2.4,0.3) rectangle (3.3,2.05);
\draw[thin,blue,rounded corners=4pt] (-2.4,-.3) rectangle (3.3,-2.05);
\draw[thin,orange,rounded corners=4pt] (5.4,-2.05) rectangle (6.21,2.05);
\draw[thin,orange,rounded corners=4pt] (7.31,-2.05) rectangle (8.13,2.05);
\draw[thin,orange,rounded corners=4pt] (8.41,-2.05) rectangle (9.23,2.05);
\draw[thin,orange,rounded corners=4pt] (10.32,-2.05) rectangle (11.14,2.05);
\end{tikzpicture}
\caption{Illustration of a DSS storing $\sm$ files, each consisting of $2 \b \sk$ symbols. The matrix $\bG_\cC$ is a generator matrix of a $[2\sn,2\sk]$ code $\cC$. }
\label{fig:DSS2}
\end{figure*}
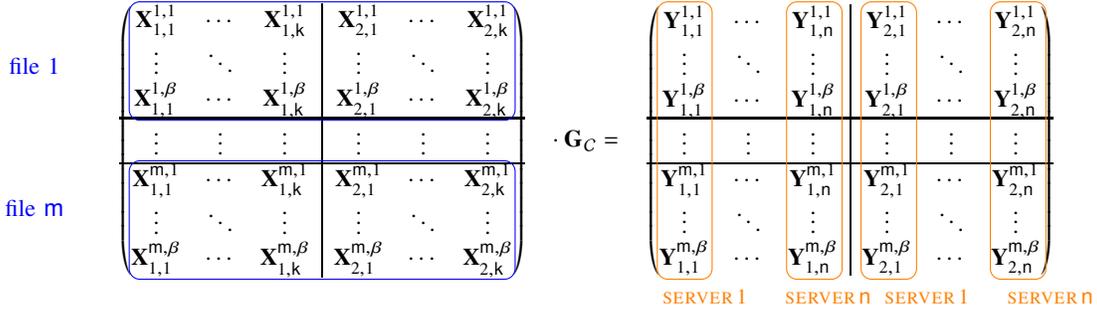

\textbf{Query and Star-Product Code.} 
Let $\st$ be the collusion parameter with $\frac{\sn}{2} \leq \sk+\st-1 < \sn$. By Lemma~\ref{lem:SelfDualStarCode} there exists an $[\sn,\st]$ GRS code $\cD'$ such that $\cS' = \cC' \star \cD'$ is an $[\sn,\sk+\st-1]$ weakly self-dual GRS code. We define the query code as the Cartesian product $\cD = \cD' \times \cD'$. Thus, for a generator matrix $\bG_{\cD'}$ of $\cD'$, the matrix $\bG_\cD = \diag{\bG_{\cD'},\bG_{\cD'}} \in \Fq^{2\st \times 2\sn}$ is a generator matrix of $\cD$.

Define $\cS = \cC \star \cD$ and $\cS' = \cC' \star \cD'$. By \eqref{eq:starCartesian} we have $\cS = \cC \star \cD = \cS' \times \cS'$, so $\cS$ is the Cartesian product of two star product codes. Define $c=d_{\cS'}-1$, where $d_{\cS'} = \sn-\sk-\st+2$ is the minimum distance of $\cS'$.

Let $\bH_{\cS'} \in \Fq^{(\sn - \sk - \st + 1) \times \sn}$ be a parity-check matrix of $\cS'$. By Definition~\ref{def:WeaklySelfDualCode}, the code $\cS'$ has a generator matrix of the form $\bG_{\cS'} = (\bH_{\cS'}^\top \ \  \bF_{\cS'}^\top)^\top$ for some $\bF_{\cS'} \in \Fq^{[2(\sk + \st - 1) - \sn] \times \sn}$. Hence, $\cS$ has a generator matrix of form
\begin{equation}
    \label{eq:StarGenMatrix}
    \bG_\cS = \ppmatrix{\diag{\bH_{\cS'},\bH_{\cS'}} \\ \diag{\bF_{\cS'},\bF_{\cS'}}} \in \Fq^{2(\sk+\st-1) \times 2\sn}.
\end{equation}

\begin{lemm}
\label{lem:MeasurementMatrix}
Let $\bG_\cS$ be the matrix defined in Eq.~\eqref{eq:StarGenMatrix} and let $\bH_\cS$ be the submatrix of $\bG_\cS$ containing its first $2(\sn-\sk-\st+1)$ rows. Let $\bw_1,\ldots,\bw_{2\sn}$ be the column vectors of $\bG_\cS$. Then, they satisfy conditions (a) and (b) of \cite[Lemma 2]{song2020colluding}, \ie
\begin{itemize}
    \item[(a)] $\bw_{\pi(1)},\ldots,\bw_{\pi(\sk+\st-1)},\bw_{\pi(1) + n},\ldots,\bw_{\pi(\sk+\st-1) + n}$ are linearly independent for any permutation $\pi$ of $[\sn]$.

    \item[(b)] $\bH_\cS \bJ^\top \bG_\cS^\top = \mathbf{0}$.
\end{itemize}
\end{lemm}

\begin{proof}
It is well-known that any subset of $\sk+\st-1$ columns of the generator matrix of an $[\sn,\sk+\st-1]$ MDS code are linearly independent. Hence, the columns $\bw_{\pi(1)},\ldots,\bw_{\pi(\sk+\st-1)}$ are linearly independent, as the first $\sn$ columns of $\bG_\cS$ generate $\cS$. The same holds for $\bw_{\pi(1) + \sn},\ldots,\bw_{\pi(\sk+\st-1) + \sn}$. Trivially, any non-zero columns of a diagonal matrix are linearly independent and property (a) follows.

Property (b) follows directly from observing that, by definition, $\bH_\cS\bG_\cS^\top = \mathbf{0}$ for any linear code with generator matrix $\bG_\cS$ and parity-check matrix $\bH_\cS$.
\end{proof}

Let $\cV$ be the space spanned by the first $2(\sn-\sk-\st+1)$ rows of $\bG_\cS$, \ie $\cV = \langle\diag{\bH_{\cS'},\bH_{\cS'}}\rangle_{\mathsf{row}}$.
By Lemma~\ref{lem:MeasurementMatrix}, 
 the space $\cV$ is self-orthogonal and the rows of $\bG_\cS$ span the space $\cV^{\perp_\mathbb{S}}$.
 Notice that $\cV$ is defined from a classical code $\cE = \langle\bH_{\cS'}\rangle_{\mathsf{row}}$, which satisfies $\cE \subset \cE^{\perp_{\mathbb{S}}}$.
 Thus, the stabilizer $\cS(\cV)$ defines a Calderbank--Steane--Shor (CSS) code \cite{Steane96,CS96}, which is defined from the self-orthogonal space $\langle\diag{\bG_{\cC_1},\bG_{\cC_2}}\rangle_{\mathsf{row}}$ with the generator matrices $\bG_{\cC_1}$ and $\bG_{\cC_2}$ of two classical codes $\cC_1$ and $\cC_2$ satisfying $\cC_1 \subset \cC_2^{\perp_{\mathbb{S}}}$.
 Our QPIR scheme will be constructed with the CSS code.

\textbf{Targeted positions.} Let $\rho = \lcm(c,\sk)/c$. Fix $\cJ = \px{1,\ldots,\max \px{c,\sk}}$ to be the set of server indices from which the user obtains the symbols of $\bY^\iota$. We consider $\cJ_1 = [c] \subseteq \cJ$ and we partition it into subsets $\cJ_1^b = \{ i + (b - 1) c/\b \mid i \in [c/\b] \},\ b \in [\b]$. Then, for $r \in [2:\rho]$ we define recursively $\cJ_r^b = \{ (j +  c/\b - 1) \; (\mathrm{mod}\ |\cJ|\ ) + 1 \mid j \in \cJ_{r-1}^b \}$ and $\cJ_r = \bigcup_{b \in [\b]} \cJ_r^b$. We will construct our scheme so that during the $r$-th iteration the user obtains the symbols $(\bY_{1,a}^{\iota,b},\bY_{2,a}^{\iota,b})$ for every $a \in \cJ_r^b$ and $b\in [\b]$.

We define
\begin{equation}
    \label{eq:MatrixN}
    \bN^{(r)} = \ppmatrix{\be_{a}^{\sn}}_{a \in \cJ_{r}}^\top \in \Fq^{c \times \sn},
\end{equation}
where $\be_a^\sn$ is the standard basis column vector of length $\sn$ with a 1 in position $a$. 
Then, the matrix $(\bG_\cS^\top \ \ (\bM^{(r)})^\top)^\top$, with $\bM^{(r)} = \diag{\bN^{(r)},\bN^{(r)}} \in \Fq^{2c \times 2\sn}$, is a basis for $\Fq^{2\sn}$. To see that this is in fact a basis observe that the row span of $\bN^{(r)}$, by definition, contains vectors of weight at most $c$. The span of $\bG_{\cS'}$ contains vectors of weight at least $ d_{\cS'}=c+1$. It follows that the spans of $\bN^{(r)}$ and $\bG_{\cS'}$ intersect trivially, which implies that their ranks add up. 

\textbf{A capacity-achieving QPIR scheme.} 
In our scheme, we use the the stabilizer formalism for the transmission of the classical files.
On the other hand, as discussed in Section~\ref{sec:stabform}, the stabilizer formalism is often used for the transmission of quantum states, which is performed by four steps of the encoding of the state, transmission over the error channel, syndrome measurement, and error-correction.
For the transmission of the classical files, similar to the QPIR scheme \cite{song2020colluding}, we construct our scheme so that the desired file is extracted by the syndrome measurement of the stabilizer code.
Then, by the same property as the superdense coding \cite{bennett1992densecoding}, our scheme can convey twice more classical information compared to the classical PIR schemes.
We refer to \cite[Section IV-B]{song2020colluding} for the detailed explanation of this idea.

Suppose the desired file is $\bX^\iota$. The queries are constructed so that the total response vector during one iteration is the sum of a codeword in $\cS$ and a vector containing $2c$ distinct symbols of $\bY^\iota$ in known locations, and zeros elsewhere.

We now describe the five steps of the capacity-achieving QPIR scheme $\Phi^\star$.

\begin{prot}\label{prot:scheme}
The first four steps are repeated in each round $r \in [\rho]$.

\begin{enumerate}
    \item \textbf{Distribution of entangled state.} Let $\cH_1,\ldots,\cH_\sn$ be $q$-dimensional quantum systems, $\sinit = q^{\sn-2(\sk+\st-1)} \cdot \bI_{q^{2(\sk+\st-1)-\sn}}$ and $\Fq^{2\sn} / \cV^{\perp_\mathbb{S}} = \px{\overline{\bw} = \bw + \cV^{\perp_\mathbb{S}} : \bw \in \langle \bM^{(r)}\rangle_{\mathsf{row}}}$. By Proposition~\ref{prop:2stab}.(b) the composite quantum system $\cH = \cH_1 \otimes \dots \otimes \cH_n$ is decomposed as $\cH = \cW \otimes \C^{q^{2(\sk+\st-1)-\sn}}$, where $\cW = \Span{\cb{\overline{\bw}} \mid \overline{\bw} \in \Fq^{2\sn} / \cV^{\perp_\mathbb{S}}}$. The state of $\cH$ is initialized as $\cb{\overline{\mathbf{0}}} \rb{\overline{\mathbf{0}}} \otimes \sinit$ and distributed such that server $s \in [\sn]$ obtains $\cH_s$.
    
    \item \textbf{Query.} The user chooses a matrix $\bZ^{(r)} \in \Fq^{\sm\b \times 2\st}$ uniformly at random. We define $\bE_{(\iota)} \in \Fq^{\sm\b \times 2c}$ with $\bE_{(\iota),p,a} = \be_{(\iota,a)}^{\sm\b},\ p\in[2],\ a\in[c]$, where $\be_{(\iota,a)}^{\sm\b}$ is the standard basis column vector of length $\sm\b$ with a 1 in coordinate $(\iota,a)$. 
    We denote by $\bQ^{(r)} \in \Fq^{\sm\b \times 2\sn}$ the matrix of all the queries, which are computed as
    \begin{equation}
        \bQ^{(r)} = \ppmatrix{\bZ^{(r)} & \bE_{(\iota)}} \cdot \ppmatrix{\bG_\cD \\ \bM^{(r)}} = \bZ^{(r)} \cdot \bG_\cD + \bE_{(\iota)} \cdot \bM^{(r)}.
        \label{eq:Queries}
    \end{equation}
    Each server $s \in [\sn]$ receives two vectors $\bQ^{(r)}_{1,s},\bQ^{(r)}_{2,s} \in \Fq^{\sm\b}$. 
    
    \item \textbf{Response.} The servers compute the dot product of each column of their stored symbols and the respective column of the queries received, \ie they compute the response $\bB^{(r)}_{p,s} = \bY_{p,s}^\top \cdot \bQ^{(r)}_{p,s} \in \Fq$, $s \in [\sn],\ p \in [2]$. 
    Server $s$ applies $\sX\paren{\bB^{(r)}_{1,s}}$ and $\sZ\paren{\bB^{(r)}_{2,s}}$ to its quantum system and sends it to the user.
    
    \item \textbf{Measurement.} The user applies the PVM $\cB^\cV = \px{\bP_{\overline{\bw}} \mid \overline{\bw} \in \Fq^{2\sn} / \cV^{\perp_\mathbb{S}}}$ on $\cH$ defined in Proposition~\ref{prop:2stab} and obtains the output $\bo^{(r)} \in \Fq^{2c}$.
    
    \item \textbf{Retrieval.} Finally, after $\rho$ rounds the user has retrieved $2\rho c = 2 \b \sk$ symbols of $\Fq$ from which he can recover the desired file $\bX^\iota$. 
\end{enumerate}
\end{prot}

\subsection{Properties of the coded QPIR scheme} \label{sec:Properties of coded QPIR}

\begin{lemm}
The scheme $\Phi^\star$ of Section~\ref{sec:scheme} is correct, \emph{i.e.}, fulfills Definition~\ref{def:correctness}.
\end{lemm}
\begin{proof}
Let us fix the round $r \in [\rho]$ and let $\bB^{(r)}$ be the vector of responses computed by the servers. By Prop.~\ref{prop:2stab}.(c)
the state after the servers' encoding is
\[
\sW(\bB^{(r)}) (\cb{\overline{\mathbf{0}}} \rb{\overline{\mathbf{0}}} \otimes \sinit) \sW(\bB^{(r)})^\dagger = \cb{\overline{\bB^{(r)}}} \rb{\overline{\bB^{(r)}}} \otimes \sinit.  
\]

We observe that $\cV^{\perp_\mathbb{S}} = \cS$ since both spaces are spanned by the rows of $\bG_\cS$. Notice that the row in coordinate $(i,b)$ of the product $\bE_{(\iota)} \cdot \bM^{(r)}$ is $\sum_{p=1}^2 \sum_{a \in \cJ_{r}^b} \d_{i,\iota} (\be_{(p,a)}^{2\sn})^\top$.  Remembering that $\be_{(p,a)}^{2\sn}$ is the standard basis column vector of length $2\sn$ with a 1 in coordinate $(p,a)$, by definition of the star product scheme the response vector is
\begin{equation}
    \label{eq:Response}
    \begin{split}
        \bB^{(r)} = & \;
        \begin{pmatrix}[c|c]
            \bB^{(r)}_{1} & \bB^{(r)}_{2}
        \end{pmatrix} = \sum_{i=1}^\sm \sum_{b=1}^{\b} \bY^{i,b} \star \bQ^{(r),i,b} \\
        = & \; \sum_{i=1}^\sm \sum_{b=1}^{\b} \paren{\bX^{i,b} \cdot \bG_\cC} \star \paren{\bZ^{(r),i,b} \cdot \bG_\cD} \\
        & + \sum_{i=1}^\sm \sum_{b=1}^{\b} \bY^{i,b} \star \Big(\sum_{a \in \cJ_r^b} \d_{i,\iota} \big(\be_{(1,a)}^{2n} + \be_{(2,a)}^{2n} \big)^\top \Big)\\
        & \; \in \cS + \sum_{b=1}^\b \sum_{a \in \cJ_r^b} \big(\bY_{1,a}^{\iota,b} \be_{(1,a)}^{2\sn} + \bY_{2,a}^{\iota,b} \be_{(2,a)}^{2\sn}\big)^\top \\
        & \hspace{10pt} = \cV^{\perp_\mathbb{S}} +
        \begin{pmatrix}[c|c]
            \bY_{1,a}^{\iota,b} & \bY_{2,a}^{\iota,b}
        \end{pmatrix}_{a \in \cJ_r^b, b \in [\b]} \cdot \bM^{(r)}.
    \end{split}
\end{equation}
The random part is encoded into a vector in $\cV^{\perp_\mathbb{S}}$ while the vector $\big( \bY_{1,a}^{\iota,b} \ | \ \bY_{2,a}^{\iota,b} \big)_{a \in \cJ_r^b, b \in [\b]} \in \Fq^{2c}$ is encoded with $\bM^{(r)}$ and hence independent of the representative of $\overline{\bo^{(r)}}$. Therefore, the user obtains the latter without error after measuring the quantum systems with the PVM $\mathcal{B}^\cV$. Recall that we fixed $\b = \lcm(c,\sk)/\sk$ for $c = d_{\cS'} - 1$. To allow the user to download exactly the desired file over $\rho$ iterations, we defined $\rho = \lcm(c,\sk)/c$. During each iteration, the user can download $2c/\b=2\sk/\rho$ symbols from each of the $\b$ rows of $\bY^\iota$, where the factor 2 is achieved by utilizing the properties of superdense coding~\cite{bennett1992densecoding}. After $\rho$ rounds the user obtained the $2\sk$ symbols $\bY^{\iota,b} \in \Fq^{2\sk}$ of each codeword corresponding to a block $\bX^{\iota,b},\ b \in [\b]$ and is therefore able to recover the file. 
\end{proof}

\begin{lemm} \label{lem:Secrecy}
The scheme $\Phi^\star$ of Section~\ref{sec:scheme} is symmetric and protects against $\st$-collusion in the sense of Definition~\ref{def:privacy}.
\end{lemm}

\begin{proof}
The idea is that user privacy is achieved since, for each subset of $\st$ servers, the corresponding joint distribution of queries is the uniform distribution over $\Fq^{\sm\b \times 2\st}$. Consider a set of $\st$ colluding servers. The set of queries these servers receive is given by $\bQ^{(r)}$ during round $r \in [\rho]$. By the MDS property of the code $\cD$ any subset of $\st$ columns of $\bG_\cD$ is linearly independent. As the columns of $\bZ^{(r)}$ are uniformly distributed and chosen independently for each $r \in [\rho]$, any subset of $\st$ columns of $\bZ^{(r)} \cdot \bG_\cD$ is statistically independent and uniformly distributed. The sum of a uniformly distributed vector and an independently chosen vector is again uniformly distributed, and therefore adding the matrix $\bE_{(\iota)} \cdot \bM^{(r)}$ does not incur any dependence between any subset of $\st$ columns and the file index~$\iota$.

For each $r \in [\rho]$, server secrecy is achieved because in every round the received state of the user is $\cb{\overline{\bB^{(r)}}} \rb{\overline{\bB^{(r)}}} \otimes \sinit$ with $\bB^{(r)} = \big( \bY_{1,a}^{\iota,b} \ | \ \bY_{2,a}^{\iota,b} \big)_{a \in \cJ_r^b, b \in [\b]}$ from \eqref{eq:Response} and this state is independent of $\bY^i$ with $i \neq \iota$.
\end{proof}

Unlike in the classical setting, the servers in the quantum setting do not need access to a source of shared randomness that is hidden from the user to achieve server secrecy. However, this should not be viewed as an inherent advantage since the servers instead share entanglement.

\begin{theo} \label{thm:rate}
The QPIR rate of the scheme in Section~\ref{sec:scheme} is
\begin{equation*}
    R(\Phi^\star) = \frac{2(\sn - \sk - \st + 1)}{\sn}
\end{equation*}
\end{theo}
\begin{proof}
The user downloads $\rho \sn$ quantum systems while retrieving $2 \sk \b \log(q)$ bits of information, thus the rate is given by 
\begin{align*}
R(\Phi^\star) &= \frac{2 k \b \log(q)}{\log(q^{\rho \sn})} \\
&= \frac{2 \rho c \log(q)}{\rho \sn \log(q)} = \frac{2(\sn - \sk - \st + 1)}{\sn}.
\end{align*}
\end{proof}

The presented scheme is an adapted version of the star-product scheme of \cite{freij2017private}, which is strongly linear \cite{Holzbaur2019ITW}. To see that the QPIR scheme is induced by this strongly linear scheme, it suffices to observe that for each $p\in[2]$ the second and third step in Protocol~\ref{prot:scheme}, up to the definition of the \emph{classical} responses $\bB_{p,s}^{(r)}$ with $s\in [\sn]$, are the same as in the star-product scheme. Hence, these steps can be viewed as two parallel instances of the star-product scheme and it follows directly from Definition~\ref{def:stronglyLinear} that this scheme is strongly linear.

\section{Converse} \label{sec:converse}

In this section, we prove Theorem~\ref{theo:2C} and Theorem~\ref{theo:conv22}. 

\subsection{Proof of Theorem~\ref{theo:2C}} \label{subsec:conv2}

Since the upper bound $1$ is trivial, we prove the quantum capacity in Theorem~\ref{theo:2C} is upper bounded by $2C_\epsilon[A]$. Let $\Phi_C$ be an arbitrary classical PIR scheme with assumptions $A$ and error probability $\epsilon$, and $\Phi_Q[\Phi_C]$ be an arbitrary dimension-squared QPIR scheme induced from $\Phi_C$ with error probability $\epsilon'$.
The PIR rate of $\Phi_C$ is upper bounded as
\begin{align}
	\frac{\sk\beta\log q}{\sum_{s=1}^{\sn} H(B_s)}
	& \le C_\epsilon[A].
\end{align}
From the definition of dimension-squared QPIR, we have $H(B_i) \le  2 \log \sd$ for all $s\in[\sn]$ for $\Phi_Q$.
Thus, the QPIR rate $R(\Phi_Q)$ is upper bounded as
	\begin{align}
	R(\Phi_Q) = \frac{\sk\beta\log q}{\sn \log \sd}
	\leq \frac{2\sk\beta\log q}{\sum_{s=1}^{\sn} H(B_s)}
	\leq 2C_\epsilon[A].
	\end{align}
Thus, the desired QPIR capacity is upper bounded by $2C_\epsilon[A]$.

\subsection{Proof of Theorem~\ref{theo:conv22}} \label{subsec:conv1}

Theorem~\ref{theo:conv22} is proved with the following idea.
If the answered state from some $\sk$ servers is independent of the targeted file $X^{\tk}$,
    the user and the remaining $(\sn-\sk)$ servers can use the answers from the $\sk$ servers as entanglement shared with the user.
Then, the entanglement-assisted classical-quantum channel capacity \cite{bennett1999} implies that the user can obtain at most $2(\sn-\sk)\log \sd$ bits of $X^\tk$, which implies Theorem~\ref{theo:conv22}.
Thus, we show that the answered state of the servers $1,\ldots,\sk$ have no information of $X^\tk$.
For the proof, we consider the process in which the $\sk$ servers apply quantum operations sequentially,
    and 
    evaluate the information of $X^\tk$ contained in the quantum systems.
Initially, the $k$ servers have quantum systems $\cH_1\otimes \cdots \otimes \cH_\sk$ and the state is independent of $X^{\tk}$.
After server $1$'s operation, the state on $\cA_1\otimes \cH_2\otimes \cdots \otimes \cH_\sk$ has at most $(\log \sd)/\sm$ bits of $X^{\tk}$ from the user secrecy.
Furthermore, we prove that as one more server applies the operation, at most $(\log \sd)/\sm$ bits of $X^{\tk}$ is added to the state of the $\sk$ servers, from the MDS-coded storage structure and the user secrecy.
Consequently, after all servers' operations, 
    the $\sk$ servers' quantum systems contain at most $(\sk\log \sd)/\sm$ bits of $X^\tk$, which converges $0$ as $\sm\to\infty$.

Throughout the proofs,
	we use superscripts $c$ (resp. $u$, $s$, $m$) over equalities and inequalities 
		for denoting they are derived from correctness (resp. user secrecy, server secrecy, MDS coded storage structure) of the QPIR scheme.
	For example, $\stackrel{\mathclap{u}}{=}$ denotes that the equality is derived from the user secrecy of QPIR scheme.

The following proofs are written with quantum mutual information and quantum relative entropy defined as follows.
When a quantum system $\cA$ has a state $\sigma = \sum_i p_i |\psi_i\rangle\langle \psi_i|$, 
the von Neumann entropy is defined as $H(\cA)_{\sigma}= \Tr {\sigma \log \sigma }= - \sum_i p_i \log p_i$.
Similar to the classical case, the mutual information and conditional mutual information are defined as $I(\cA;\cB)_\sigma = H(\cA)_\sigma +H(\cB)_\sigma - H(\cA\cB)_\sigma$ and
 $I(\cA;\cB|\cC)_\sigma = I(\cA;\cB\cC)_\sigma -I(\cA;\cC)_\sigma $, respectively.
For two states $\sigma$ and $\sigma'$ on $\cA$,
 the quantum relative entropy is defined as
$D(\sigma\|\sigma') = \Tr {\sigma (\log \sigma - \log \sigma')}$.
Similar to classical case, we have ${I(\cA ; \cB)_{\sigma}
	= D( \sigma \| \sigma_{\cA} \otimes \sigma_{\cB} )} $

For the proof, we prepare two propositions.

\begin{prop}[Fano's inequality] \label{prop:fano}
Let $X,Y$ be random variables with values in $[n]$
    and $Z$ be any random variable.
Then, 
$H(X|YZ) \le  \epsilon \log n + h_2(\epsilon)$,
where $\varepsilon = \Pr[X\neq Y]$.
\end{prop}

\begin{prop} \label{prop:mi_inf}
Let $\kappa$ be a CPTP map from $\cA$ to $\cB$
	and $\sigma$ be a state on $\cA\otimes \cC$.
Then, 
$I(\cA ; \cC)_{\sigma} \geq I(\cB; \cC)_{\kappa \otimes \id_{\cC}(\sigma)}$,
where $\id_{\cC}$ is the identity operator on $\cC$.
\end{prop}

\begin{proof}
The proposition follows from the following inequality 
\begin{align*}
\lefteqn{I(\cA ; \cC)_{\sigma}
	= D( \sigma \| \sigma_{\cA} \otimes \sigma_{\cC} )} \\
	&\geq D(\kappa \otimes \id_{\cC}(\sigma) \| \kappa(\sigma_{\cA}) \otimes \sigma_{\cC} )
	= I(\cB; \cC)_{\kappa \otimes \id_{\cC}(\sigma)},
\end{align*}
where $\sigma_{\cA}$ and $\sigma_{\cC}$ are reduced states on $\cA$ and $\cC$,
	and the inequality is from the data-processing inequality of the quantum relative entropy.
\end{proof}

Theorem~\ref{theo:conv22} is proved by the following two lemmas.
\begin{lemm} \label{lemm:onefileupper}
The size of one file is upper bounded as 
\begin{align}
 \sk\beta \log q \leq 
	\frac{2(\sn-\sk) \log \sd +  I(\cA_{[\sk]} ;  X^{ \tk} | Q^\tk ) + h_2(\epsilon)}
		{1-\epsilon},
\end{align}
where $\epsilon = \max_{\tk\in[\sm]} \Pr [ X^{\tk} \neq \tilde{X}^{\tk}]$.
\end{lemm}

\begin{proof}
Fix the index of the targeted file as $\tK = \tk = \argmax_{\tk\in[\sm]} \Pr [ X^{\tk} \neq \tilde{X}^{\tk}]$.
The uniformity of $X^{\tk}\in\Fq^{\beta\times\sk}$ and the Fano's inequality (Proposition~\ref{prop:fano}) imply
\begin{align}
&I(\hat{X}^{\tk} ; X^{\tk}| Q^\tk)
         = H(X^{\tk}|Q^\tk) - H(X^{\tk}|\hat{X}^{\tk}Q^\tk) \\
 &\ge (1-\epsilon) \sk\beta \log q - h_2(\epsilon).
 \label{eq:fano1}
\end{align}
From Proposition~\ref{prop:mi_inf},
    the mutual information in the above inequality is upper bounded as 
\begin{align}
    I(\cA ;  X^{ \tk} | Q^\tk )
    \geq 
    I(\hat{X}^{\tk} ; X^{\tk}| Q^\tk)
    .
    \label{eq:prop6-2}
\end{align}
Furthermore, the left-hand side of the above inequality is upper bounded as 
\begin{align*}
	I(\cA ;  X^{ \tk} | Q^\tk )
			&= I(\cA_{[\sk+1:\sn]} ;  X^{ \tk} | \cA_{[\sk]} Q^\tk )
			 + I(\cA_{[\sk]} ;  X^{ \tk} | Q^\tk )\\
			&\le 2 \log \dim \cA_{[\sk+1:\sn]}
			 + I(\cA_{[\sk]} ;  X^{ \tk} | Q^\tk )\\
			& = 2 (\sn-\sk) \log \sd
			 + I(\cA_{[\sk]} ;  X^{ \tk} | Q^\tk ).
 \label{eq:6-1-3}
\end{align*}
Thus, combining \eqref{eq:fano1}, \eqref{eq:prop6-2}, and \eqref{eq:6-1-3}, we obtain the desired lemma.
\end{proof}

\begin{lemm} \label{lemm:mutinf}
$\lim_{\sm \to \infty} I(\cA_{[\sk]} ;  X^{\tk} | Q^\tk ) = 0.$
\end{lemm}

With Lemmas~\ref{lemm:onefileupper} and \ref{lemm:mutinf}, we prove Theorem~\ref{theo:conv22} as follows.
From Lemma~\ref{lemm:onefileupper}, the $[\sn,\sk,1]$-QPIR capacity is upper bounded as 
\begin{align}
& C_{\sm,\epsilon}^{[\sn,\sk,1]} 
    = 
	\sup \frac{\sk\beta \log q}{\sn\log\sd} \\
&	\leq 
	\frac{1}{1-\epsilon}\paren*{
	\frac{2(\sn-\sk)}{\sn} 
	+ \frac{ I(\cA_{[\sk]} ;  X^{ \tk} | Q^\tk ) + h_2(\epsilon)}{\sn\log\sd}
	}.
	\label{up1}
\end{align}
Furthermore, 
	Lemma~\ref{lemm:mutinf} proves that 
	$I(\cA_{[\sk]} ;  X^{ \tk} | Q^\tk )$ approaches zero as the number of files $\sm$ goes to infinity,
	and $h_2(\epsilon) \to 0$ as $\epsilon \to 0$.
Thus, as $\sm\to\infty$ and $\epsilon\to0$, 
    the capacity is upper bounded by $2(1-\sk/\sn)$,
        which implies Theorem~\ref{theo:conv22}.

In the remainder of this subsection, we prove Lemma~\ref{lemm:mutinf}.
For the proof, we prepare the following lemma.

\begin{lemm} \label{lemm:userinf}
Suppose that $t\in[\sn]$ and $\cT\subset[\sn]$ satisfy $t \not \in \cT$.
Then, 
\begin{align}
	I(\cA_{t} \cH_{\cT} ; Y_{t}^{\tk} | Q^\tk) \leq  \frac{2\log \sd}{\sm} .
	\label{34}
\end{align}
\end{lemm}

\begin{proof}
Since the operation from $\cH_{t}$ to $\cA_{t}$ is applied on the quantum system of dimension of $\sd$,
	we have
\begin{align}
I(\cA_{t} \cH_{\cT} ; Y_{t} | Q^\tk) \leq 2 \log \sd. 
	\label{eq:ll-1}
\end{align}
On the other hand, we have
\begin{align}
&I(\cA_{t} \cH_{\cT} ; Y_{t} | Q^\tk) 
	= \sum_{j=1}^{\sm} I(\cA_{t} \cH_{\cT} ; Y_{t}^j | Y_{t}^{[j-1]} Q^\tk) \\
	&= \sum_{j=1}^{\sm} I(\cA_{t} \cH_{\cT} Y_{t}^{[j-1]} ; Y_{t}^j |  Q^\tk) 
	\geq \sum_{j=1}^{\sm} I(\cA_{t} \cH_{\cT}  ; Y_{t}^j |  Q^\tk) \\
	&\stackrel{\mathclap{u}}{=} \sm I(\cA_{t} \cH_{\cT}  ; Y_{t}^{\tk} |  Q^\tk) ,
	\label{eq:ll-2}
\end{align}
where the last equality follows from the user secrecy condition.
Thus, combining \eqref{eq:ll-1} and \eqref{eq:ll-2},
	we obtain the desired inequality \eqref{34}.
\end{proof}

Now, we prove Lemma~\ref{lemm:mutinf}.
\begin{proof}[Proof of Lemma~\ref{lemm:mutinf}] 
By mathematical induction, we prove
	\begin{align}
	\lim_{\sm\to\infty}
	I(\cA_{[j]} \cH_{[j+1:\sk]} ; Y_{[j]}^{\tk} | Q^\tk) = 0 
	\end{align}
	for any $j \in [\sk]$.
Then, the case for $j =\sk$ proves the lemma.

First, the case $j=1$ follows from Lemma~\ref{lemm:userinf}.
Next, assuming 
	\begin{align}
	\lim_{\sm\to\infty}
	I(\cA_{[j]} \cH_{[j+1:\sk]} ; Y_{[j]}^{\tk} | Q^\tk) = 0,
	\label{eq:assumplemm}
	\end{align}
we prove 
	\begin{align}
	\lim_{\sm\to\infty}
	I(\cA_{[j+1]} \cH_{[j+2:\sk]} ; Y_{[j+1]}^{\tk} | Q^\tk) = 0
	\end{align}
for $j\in[\sk-1]$.
Since 
	\begin{align}
	& I(\cA_{[j+1]} \cH_{[j+2:\sk]} ; Y_{[j+1]}^{\tk} | Q^\tk)\\
	&= 
		I(\cA_{[j+1]} \cH_{[j+2:\sk]} ; Y_{[j]}^{\tk} | Q^\tk)
		+
		I(\cA_{[j+1]} \cH_{[j+2:\sk]} ; Y_{j+1}^{\tk} | Y_{[j]}^{\tk}Q^\tk),
		\label{eq:twoterms}
	\end{align}
we prove that the two terms of \eqref{eq:twoterms} approaches $0$ as $\sm\to\infty$.
Then, we obtain the desired statement by induction.

The first term of \eqref{eq:twoterms} is upper bounded as 
	\begin{align*}
	&I(\cA_{[j+1]} \cH_{[j+2:\sk]} ; Y_{[j]}^{\tk} | Q^\tk) 
	\leq I(\cA_{[j+1]} \cH_{[j+2:\sk]} Y_{j+1}; Y_{[j]}^{\tk} | Q^\tk)   \nonumber  \\
	&\stackrel{\mathclap{(a)}}{\leq} I(\cA_{[j]} \cH_{[j+1:\sk]} Y_{j+1}; Y_{[j]}^{\tk} | Q^\tk) 
	\stackrel{\mathclap{m}}{=} I(\cA_{[j]} \cH_{[j+1:\sk]} ; Y_{[j]}^{\tk} | Q^\tk)  \label{eq:indep111},
	\end{align*}
where $(a)$ follows from Proposition~\ref{prop:mi_inf}
    and 
	the last equality holds because $Y_{j+1}$ is independent of all other quantum systems and random variables.
Thus, by the assumption \eqref{eq:assumplemm},
	the first term of \eqref{eq:twoterms} approaches $0$ as $\sm\to\infty$.

The second term of \eqref{eq:twoterms} is upper bounded as
	\begin{align}
	& I(\cA_{[j+1]} \cH_{[j+2:\sk]} ; Y_{j+1}^{\tk} | Y_{[j]}^{\tk}Q^\tk) \\
	&\stackrel{\mathclap{m}}{=} I(\cA_{[j+1]} \cH_{[j+2:\sk]} Y_{[j]}^{\tk} ; Y_{j+1}^{\tk} | Q^\tk)\\
	&\leq  I(\cA_{[j+1]} \cH_{[j+2:\sk]} Y_{[j]} ; Y_{j+1}^{\tk} | Q^\tk)\\
	&\leq  I(\cA_{j+1} \cH_{[\sk]\setminus\{j+1\}} Y_{[j]} ; Y_{j+1}^{\tk} | Q^\tk)
		\label{eq:unit111}
		\\
	&\stackrel{\mathclap{m}}{=}  I(\cA_{j+1} \cH_{[\sk]\setminus\{j+1\}} ; Y_{j+1}^{\tk} | Q^\tk)
	\leq  \frac{\log \sd}{\sm},
	\end{align}
where 
	\eqref{eq:unit111} follows from Proposition~\ref{prop:mi_inf}
	and
	the last inequality is from Lemma~\ref{lemm:userinf}.
Thus, the second term of \eqref{eq:twoterms} approaches $0$ as $\sm\to\infty$.
\end{proof}

\section{Conclusion} \label{sec:conclusion}

In this paper, we have studied the capacity of QPIR/QSPIR with $[\sn,\sk]$-MDS coded storage and $\st$ colluding servers.
As general classes of QPIR, we defined stabilizer QPIR and dimension-squared QPIR induced from classical strongly linear PIR.
We have proved that the capacities of stabilizer QPIR/QSPIR and dimension-squared QPIR/QSPIR induced from strongly linear PIR are $2(\sn-\sk-\st+1)/\sn$.
When there is no collusion, \ie $\st =1$, we have proved that the asymptotic capacity of QPIR/QSPIR is $2(\sn-\sk)/\sn$, when the number of files $\sm$ approaches infinity.
These capacities are greater than the known classical counterparts.
For the achievability, we have proposed a capacity-achieving QSPIR scheme.
The proposed scheme combined 
    the star product PIR scheme \cite{freij2017private}
    and the QPIR scheme with the stabilizer formalism \cite{song2020colluding}.

As open problems, we state three directions for extending our results.
The first direction is to find the general capacity of QPIR/QSPIR with MDS coded storage and colluding servers.
This problem in full generality is also unsolved in the classical setting. Partial solutions were given in \cite{Sun2018conjecture,holzbaur2019capacity}, which imply that the combination of collusion and coded storage leads to involved linear dependencies that need to be taken into account for a general converse proof.
Note that as the capacities proved in these works depend on the number of files $\sm$, it is possible that they exceed the asymptotic QPIR capacity proved in this work for a very small number of files.

The second direction is to find non-stabilizer QPIR schemes.
Most of the existing multi-server QPIR schemes are stabilizer QPIR schemes.
Finding non-stabilizer QPIR schemes is the first step towards the achievability part of the general non-asymptotic capacity theorem.

The third direction is to clarify the trade-off between the amount of entanglement and the capacity. 
However, even in the case of only two servers, it is very challenging to derive the capacity with restricted entanglement.
As a related study, the entanglement-assisted classical capacity for a noisy quantum channel \cite{wh2020} has been recently studied with several new techniques.

\bibliographystyle{IEEEtran}
\bibliography{main}

\end{document}